\documentclass{svjour3}
\smartqed  

\usepackage{varwidth,comment}

\usepackage[toc,page]{appendix}

\usepackage{amsmath,pgfplots,caption}
\usepackage{amssymb}
\usepackage{amsfonts}
\usepackage{graphicx,epstopdf}
\usepackage{mathrsfs,mathtools}
\usepackage{color}
\usepackage{subfigure,bbm,url}
\usepackage[numbers,sort]{natbib}

\usepackage{algorithmicx}
\usepackage{algorithm}
\usepackage[noend]{algpseudocode}
\makeatletter
\def\BState{\State\hskip-\ALG@thistlm}
\makeatother

\newcommand{\bigslant}[2]{{\raisebox{.2em}{$#1$}\left/\raisebox{-.2em}{$#2$}\right.}}
\def\Id{{\mathbb{I}}}

\newcommand\restr[2]{{
  \left.\kern-\nulldelimiterspace 
  #1 
  \vphantom{\big|} 
  \right|_{#2} 
  }}

\def\ZZ{{\mathbb{Z}}} 
\def\QQ{{\mathbb{Q}}} 
\def\RR{{\mathbb{R}}} 
\def\CC{{\mathbb{C}}} 

\def\NN{{\mathbb{N}}}

\def\sing{\text{sing}}
\def\crit{\text{crit}}
\def\reg{\text{reg}}
\def\scC{{\mathscr{C}}}
\def\scS{{\mathscr{S}}}
\def\sfP{{\mathsf{P}}}

\def\softO{\ensuremath{{\mathcal{O}}{\,\tilde{ }}}}
\def\bigO{\ensuremath{{\mathcal{O}}}}

\def\sfG{{\mathsf{G}}} 
\def\sfH{{\mathsf{H}}} 

\def\GL{{\mathrm{GL}}} 
\def\rank{\mathrm{rank}} 
\def\jac{{D}}
\def\minors{\text{minors}}

\def\X{{x}} 
\def\Y{{y}} 

\def\x{{x}} 
\def\y{{y}} 
\def\fiber{{t}}

\def\setD{{\mathcal{D}}}
\def\setV{{\mathcal{V}}}

\def\setZ{{\mathcal{Z}}}

\def\cc{{\mathcal{C}}}

\def\zarA{{\mathscr{A}}}
\def\zarU{{\mathscr{U}}}

\def\zarM{{\mathscr{M}}}

\def\zarW{{\mathscr{W}}}

\def\zarfiber{{\mathscr{T}}}

\newcommand{\zeroset}[1]{{\mathcal{Z}(#1)}}
\newcommand{\ideal}[1]{{I(#1)}}

\newcommand{\new}[1]{{{#1}}}                     

\newcommand{\rect}[1]{{#1}} 
\newcommand{\mynew}[1]{{{#1}}}

\pgfplotsset{compat=1.9}

\title{Real root finding for low rank linear matrices}

\begin{document}

\author{Didier Henrion \and Simone Naldi \and Mohab {Safey El Din}}

\institute{D. Henrion \at
  CNRS, LAAS, 7 avenue du colonel Roche, F-31400 Toulouse; France.\\
  Universit\'e de Toulouse; LAAS, F-31400 Toulouse, France.\\
  Faculty of Electrical Engineering, Czech Technical University in Prague, Czech Republic.
  \email{henrion@laas.fr} \\
  \and
  S. Naldi \at
  Univ. Limoges, CNRS, XLIM, UMR 7252, F-87000 Limoges, France.\\
  \email{simone.naldi@unilim.fr}\\
  \and
  M. {Safey El Din} \at
  Sorbonne Universit\'e, CNRS, Inria, Laboratoire d'Informatique de Paris 6,
  LIP6, Equipe PolSys, F-75252, Paris, France.\\
  \email{Mohab.Safey@lip6.fr}
}


\date{\today}

\maketitle

\begin{abstract}
  We consider $m \times s$ matrices (with $m\geq s$) in a real affine
  subspace of dimension $n$. The problem of finding elements of low rank
  in such spaces finds many
  applications in information and systems theory, where low rank is
  synonymous of structure and parsimony. We design computer algebra
  algorithms, based on advanced methods for polynomial system solving,
  to solve this problem efficiently and exactly: the input are the
  rational coefficients of the matrices spanning the affine subspace
  as well as the expected maximum rank, and the output is a rational
  parametrization encoding a finite set of points that intersects each
  connected component of the low rank real algebraic set. The
  complexity of our algorithm is studied thoroughly. It is 
  polynomial in $\binom{n+m(s-r)}{n}$. It improves on the
  state-of-the-art in computer algebra and effective real algebraic
  geometry. Moreover, computer experiments show the practical
  efficiency of our approach.  \keywords{{Symbolic Computation; Low
      rank matrices; Polynomial system solving; Real algebraic
      geometry.}}  \subclass{13-XX \and 14Q20 \and 12Y05 \and 68W30}
\end{abstract}

\section{Introduction}



\subsection{Problem statement}\label{statement}

Let $\QQ$, $\RR$ and $\CC$ be respectively the fields of
rational, real and complex numbers. Let \rect{$s$}, $m$, $n$, $r$
be positive integers with $0 \leq r {< s \leq m}$ and let $A_0,
\ldots, A_n$ be \rect{$m \times s$} matrices with entries in $\QQ$.
Let $x=(x_1,\ldots, x_n)$ be variables.
We consider the {\it linear matrix} $A(x)$ defined by
$$
 (\x_1, \ldots, \x_n) \mapsto A(x) = A_0+x_1 A_1+\cdots+x_n A_n.
$$
{By abuse of notation, we denote the vector $(A_0,A_1,\ldots,A_n) \in
(\QQ^{\rect{m\times s}})^{n+1}$ by $A$.}
Given $A$ as above, the set
$$
{\setD}_r = \big\{ x \in \CC^n \ \big| \ \rank \ A(x) \leq r \big\}.
$$
is called a determinantal variety. The goal of this paper is to design
an efficient algorithm for {deciding the emptiness of $\setD_r \cap \RR^n$
  and, if it is not empty, for computing at least one point in each
  connected component of $\setD_r\cap \RR^n$.}


Our algorithm is symbolic, that is its output is an exact encoding of finitely
many points whose coordinates are algebraic numbers, given by a rational
pa\-ra\-me\-tri\-za\-tion with coefficients in $\QQ$. {This is a vector
  $q=(q_0,q_1,\ldots,q_{n+1}) \in \QQ[t]^{n+2}$ of univariate polynomials
  such that
  {$q_{n+1}$ is square-free, $q_0={\partial q_{n+1}}/{\partial t}$
    and $\deg(q_i)<\deg(q_{n+1})$ for $1\leq i \leq n$}, and such that
  the following set
\begin{equation}
\label{ratpar} 
\setZ = \left\{ \left(\frac{q_1(t)}{q_0(t)},\cdots,\frac{q_n(t)}{q_0(t)}\right) \in \CC^n \: :\: q_{n+1}(t)=0 \right\}.
\end{equation}
is contained in $\setD_r$ and contains at least one point in each
connected component of $\setD_r \cap \RR^n$.
Such an encoding for finite algebraic sets goes back to the work of
Macaulay and Kronecker \cite{Kronecker82, Macaulay16} and has been extensively
used and developed for computer algebra methods for solving polynomial
systems (see e.g. \cite{GiMo89, Rouillier99, Lecerf2000}).

 


\subsection{Motivations} \label{motiv}

{The problem of finding low rank matrices in a given affine space
has many applications in systems, signal and information engineering, where
low rank typically corresponds to sparsity and structure requirements.
For instance, in the context of semi\-de\-finite
programming (SDP) hierarchies for polynomial optimization
\cite{l10}, low rank moment matrices provide guarantees
of global optimality of a convex relaxation of a non-convex 
optimization problem.}
Similarly, the geometry of low rank structured matrices
  (e.g. Hurwitz, Hankel, Toeplitz,
  resultant matrices) is pervasive in algebraic approaches to
  information engineering (including systems control, signal
  processing, computer vision).
{In these cases,
  the given affine subspace lies in the linear space of symmetric (or more
  structured) matrices, while in this paper we address the problem from a more
  general point of view.}
\begin{example}
  {The paper \cite{kileeldistortion} studies distortion varieties, special
    algebraic varieties arising from computer vision. These have
    determinantal structure. For instance, the distortion variety
    in \cite[Ex. 1.1]{kileeldistortion} is the zero set of all $2 \times 2$ minors
    of the matrix
    $$
    {A}(x_1,\ldots,x_{11}) =
    \left(
    \begin{array}{cccccc}
      x_1 & x_2 & x_3 & x_4 & x_5 & x_6 \\
      x_7 & x_8 & x_9 & x_{10} & x_6 & x_{11} \\
    \end{array}
    \right).
    $$
    The matrix $A$ defines a linear space of co-dimension 1 in the
    space of $2 \times 6$ matrices (indeed the $(1,6)-$entry equals
    the $(2,5)-$entry), and we are interested in the locus of rank-one
    matrices of this form. The real points of this variety can be
    sampled with the algorithm developed in this paper, since the
    regularity assumptions needed by our algorithm (defined in Section
    \ref{genass}) are satisfied by this example.}
\end{example}
The specific geometry of low rank manifolds can be exploited to design
efficient nonlinear local optimization algorithms
\cite{absil2008}.
     {Linear matrices
  and their loci of rank defects are the object of the so-called low
  rank approximation problem \cite{OttSpaStu14} and model problems arising in
  medical imagery \cite{BFJSV16}.}

In our paper, we are not after trying to solve approximately
large-scale problem instances with floating point arithmetic. In
contrast, our focus is on symbolic computation and rigorous
algorithms.  This means that we are not concerned with numerical
scaling and conditioning issues. 
We provide mathematical guarantees of exactness of the output of
our algorithm, under the assumption that the input is also exactly
provided in rational arithmetic and satisfies some algebraic
assumptions that are specified below. Obviously, these guarantees come
with a price, and our algorithm complexity is exponential in the
number of variables or problem size, and hence limited to ``small''
dimensions.  But this is not specific to our algorithm, this
limitation is shared with all symbolic computation methods: our
algorithm should be applied to small-size problems for which it is
absolutely crucial to find exact solutions.

{Still, in this context, complexity issues are topical and can lead to practical
improvements.} The main difference with the state-of-the-art is that the
complexity achieved by our algorithm is essentially quadratic in a
multilinear B\'ezout bound on the maximum number of complex solutions
encoded by the output. This bound is itself dominated by
\rect{${\binom{n+m(s-r)}{n}}^3$}. Hence, for particular sub-classes of the
problem, for example when the
\rect{maximum dimension} of the matrix is fixed, the
multilinear bounds (and hence the complexity) are polynomial in the
number of variables. We will see that this leads to faster implementations at
the end of the paper.



\subsection{State of the art}


We distinguish in the state-of-the-art three subproblems. The first
one is on computing sample points in each connected component of real
algebraic sets, hence without taking care of the determinantal
structure we consider here. Next, we review on previous work taking
care of the determinantal structure but in the context of
zero-dimensional algebraic sets. Finally, we consider real algebraic
sets defined by rank constraints on matrices with polynomial entries.

Computing real solutions of systems of $n$-variate polynomial equations is
a central question
in computational geometry and {effective real algebraic
  geometry}. Since one typically deals with positive dimensional
solution sets, one possible approach is to design algorithms computing
a finite set intersecting each connected component {of the real
  solution set}. The complexity of Tarski's algorithm \cite{Tarski} was
not elementary recursive, while Collins's Cylindrical Algebraic Decomposition
\cite{c-qe-1975} is doubly exponential in $n$. Since Thom-Milnor bound
for the maximum number of connected components of a real algebraic set
\cite[Theorem 7.23]{BaPoRo06} is singly exponential in $n$, tremendous
efforts have been made to obtain optimal complexities.

Grigoriev and Vorobjov introduced in \cite{GV88} the critical point
method which culminates with the algorithms in \cite[Ch.13]{BaPoRo06}
running in time singly exponential in $n$. The algorithms in \cite{BGHM1,SaSc03}
also rely on the computation of critical points: On inputs of degree $\leq d$,
they lead to complexities which are essentially cubic (resp. quartic) in $d^n$
for the smooth (resp. singular) case. These techniques have also been used in
the context of polynomial optimization \cite{GS14}.

Dedicated algorithms in e.g. \cite{degeneracyloci,GiLeSa01} or
Gr\"obner bases \cite{FauSafSpa13} can be used to compute generic
points in algebraic varieties in presence of determinantal structure.
Observe that, by the way, these are not sufficient to be applied to the
problem of real root finding for positive dimensional real algebraic sets.
In the context of {determinantal varieties} $\setD_r$, the following cases
have been already treated:
\begin{itemize}
\item \rect{$m=s$ and} $r=s-1$: {in \cite{HNS2014}, we designed
    a dedicated algorithm for computing sample points in each
    connected component of the studied real algebraic set under some
    genericity assumption on the input matrix pencil;}
\item {$m=s$ and the considered matrix is {\it symmetric}: we
    designed in \cite{HNS16} a dedicated algorithm for this situation
    without any other constraint on $r$ than $r\leq m-1$, again under
    some genericity assumption on the input matrix pencil. In
    \cite{HNS2015a}, we also tackle the situation where the linear
    matrix is Hankel.}
\end{itemize}

{Observe that the cases $m\neq s$ and arbitrary $r$ were
  pending.  In the current paper, we deal with the case $m\neq s$
  (assuming without loss of generality that $m\geq s$) and arbitrary
  $r\leq m-1$. This paper builds on the previous work \cite{HNS2014}:
  the spirit and the statement of our main result is rather close to
  this previous work but many of the techniques used in \cite{HNS2014}
  cannot be applied {\it mutatis mutandis} to the more general setting
  we consider here and need to be adapted and generalized.}

  {Moreover, we highlight that, contrarily to our previous contribution \cite{HNS2014} concerning determinantal hypersurfaces, in this paper we explicitely describe the dependencies between the choice of parameters during the main algorithm. }




\subsection{Paper outline}

The algorithm described in this paper, with input a \rect{$m \times s$} linear matrix $A(x) = A_0+x_1A_1+\ldots+x_nA_n$, with {$m\geq s$,} $A_i \in \QQ^{\rect{m \times s}}$, $i=0,1, \ldots, n$, and an integer $r \leq s-1$, computes a rational parametrization of a finite set intersecting each connected component of $\setD_r \cap \RR^n$. The design of the algorithm is intended to take advantage of the special structure of the input problem and hence to behave better than algorithms based on the critical point method that solve the same problem in a more general setting.

{ Since algebraic sets defined by minors of fixed size of a
  polynomial matrix are generically singular, the input of our
  algorithm does not satisfy regularity properties. Hence, the first
  step is to generate a second algebraic set $\setV_r$, defined by
  quadratic equations $A(x) Y(y)=0$, where $Y(y)$ is a rectangular
  matrix whose columns generate the kernel of $A(x)$. The set we have
  obtained is a lifting of $\setD_r$, which is traditionally called an
  {\it incidence variety}.  }

We investigate properties of this incidence variety, {proving that
  unlike $\setD_r$, the lifted set $\setV_r$ is regular (smooth and
  equidimensional) when the input matrices $A=(A_0,A_1,\ldots,A_n)$
  lie outside a given algebraic hypersurface in
  $(\QQ^{\rect{m \times s}})^{n+1}$.} We show that our problem can be
reduced to compute finitely many critical points of the restriction of
a general linear projection to this lifted set. The system that
defines these critical points has a special sparsity structure,
{namely it is bilinear in three groups of variables (the variables
  $x$ describing $\setD_r$, the variables $y$ encoding the kernel, and
  Lagrange multipliers $z$).} Using the symbolic homotopy algorithm in
\cite{SaSc17} (which builts upon the one in
\cite{jeronimo2009deformation}), one can compute a rational
parametrization of these critical points by exploiting this sparsity
structure. We establish a bound $\delta$ on the degree of the
parametrization, and, using \cite{SaSc17}, we show
that the complexity is essentially quadratic on $\delta$. This bound
is dominated by \rect{${\binom{n+m(s-r)}{n}}^3$}.  {Note that this
  complexity estimate does not take into account the cost of checking
  that the genericity assumption on the input is satisfied, which we
  suppose to be true.}

  Moreover, we provide computer experiments that show that our
  strategy allows to tackle problems that are unreachable by
  implementations of other generic algorithms based on the critical
  point method.



\section{Definitions and notation}\label{sec:notations}

\subsection{Basic notions}


We denote by $\QQ^n$ (resp. $\CC^n$) the set of vectors of
length $n$ with entries in $\QQ$ (resp. $\CC$).
A subset $\setV \subset \CC^n$ is an affine algebraic variety
(equivalently affine algebraic set) defined over $\QQ$ if it is the
common zero locus of a system of polynomials
$f = (f_1, \ldots, f_q) \in \QQ[x]^q$, with $x=(x_1, \ldots, x_n)$. We
also write $\setV = f^{-1}(0) = \zeroset{f}$. Algebraic varieties in
$\CC^n$ define the closed sets of the so-called Zariski
topology. Zariski open subsets of $\CC^n$ are sets whose complement
are Zariski closed; they are either empty or dense in $\CC^n$.

The set of all polynomials vanishing on an algebraic set $\setV$ is an
ideal and it is denoted by $\ideal{\setV} \subset \QQ[x]$.
This ideal is radical (i.e. $g^k \in \ideal{\setV}$ for some integer $k$
implies that $g \in \ideal{\setV}$) and it is generated by a finite set
of polynomials, say $f=(f_1, \ldots, f_p)$. We also write
$\ideal{\setV} = \langle f_1,\ldots,f_p \rangle = \langle f \rangle$
when a set of generators is known.

Let $\GL_n(\CC)$ (resp. $\GL_n(\QQ)$) be the set of non-singular
$n \times n$ matrices with entries in $\CC$ (resp. $\QQ$). The identity
matrix is denoted by $\Id_n$. Given a matrix $M \in \GL_n(\QQ)$ and a
polynomial system $\X \in \CC^n \mapsto f(\X) \in \CC^p$ we denote by
$f \circ M$ the polynomial system $\X \in \CC^n \mapsto f(M\, \X) \in \CC^p$.
If $\setV = \zeroset{f}$, the image set $\zeroset{f \circ M}=\{\X \in \CC^n : f(M\X) = 0\}
=\{M^{-1}\X \in \CC^n : f(\X) = 0\}$ is denoted by
$M^{-1}\setV$. Given $q \leq n$ and $M \in \CC^{m \times m}$, we denote by
$\minors(q, M)$ the set of determinants of $q \times q$ submatrices of
$M$. {The transpose of a matrix $M$ is denoted by $M^T$.}

For $f \in \QQ[x]^q$, we denote by $\jac f$ the Jacobian matrix of
$f$, that is the $q \times n$ matrix $\jac f = (\frac{\partial
  f_i}{\partial \X_j})_{i,j}$.  When $f$ generates a radical ideal,
the codimension $c$ of $\zeroset{f}$ is the maximum rank of $\jac f$
evaluated at points in $\zeroset{f}$. Its dimension is $n-c$. The
algebraic set $\setV=\zeroset{f}$ is said irreducible, if it is not
the union of two algebraic sets strictly contained in
$\zeroset{f}$. If $\setV$ is not irreducible, it is decomposable as
the finite union of irreducible algebraic sets, called the irreducible
components.  If all the irreducible components have the same
dimension, $\setV$ is equidimensional.  The dimension of $\setV$
coincides with the maximum of the dimensions of its components.

Let $f \colon \CC^n \to \CC^q$ generate a radical ideal, and let $\setV = \zeroset{f}$ be
equidimensional of dimension $d$. A point $\x \in \setV$ such that the
rank of $\jac f$ is equal to $n-d$ is a regular point, otherwise is a
singular point. We denote by $\reg\:\setV$ and $\sing\:\setV$
respectively the set of regular and singular points of $\setV$.

Let $f \colon \CC^n \to \CC^q$ generate a radical ideal, and let $\setV = \zeroset{f}$ be
equidimensional of dimension $d$. Let $g \colon \CC^n \to \CC^p$. {A point
$x \in \reg\:\setV$ is a critical point of the restriction of $g$ to $\setV$
if the minors of size $n-d+p$ of the extended Jacobian
matrix $\jac (f,g)$ vanish at $x$.}
The Zariski-closure of the set of critical points is denoted by $\crit(g,\setV)$.
{Let $\pi_1 \colon \CC^{n} \to \CC$
  be the projection $\pi_1(x)=x_1$ and let $\jac_1 f$ be the matrix
  obtained by deleting the first column of $\jac f$. Then
  $\crit(\pi_1,\setV)$ is equivalently defined by the zero set of
  {the polynomials in} $f$ and the maximal minors of
  $\jac_1 f$. {The set $\crit(\pi_1,\setV)$ is also called a {\it
      polar variety} when $\setV$ is smooth.  }

\subsection{Incidence variety} \label{ssec:smoothness:incvar}

Let $A=(A_0, A_1, \ldots, A_n)$ be
\rect{$m \times s$} matrices {($m\geq s$)} with entries in
$\QQ$, and $A(x) = A_0 + x_1 A_1 + \cdots + x_n A_n$ the associated
linear matrix.  If
$\x \in \setD_r = \{ x \in \CC^n : \rank \, A(x) \leq r\}$, the
\rect{right} kernel of $A(x)$ \rect{is a subspace of} dimension
$\geq s-r$ \rect{in $\CC^s$} by linear algebra.

We introduce
\rect{$s(s-r)$} variables
$y=(y_{1,1}, \ldots, y_{\rect{s,s-r}})$, stored in a \rect{$s \times (s-r)$} linear matrix
\[
Y(y)=
\left(
\begin{array}{ccc}
y_{1,1} &  \cdots  & y_{1,s-r} \\
\vdots &         & \vdots  \\
\vdots &         & \vdots  \\
\rect{y_{s,1}} & \cdots & \rect{y_{s,s-r}}
\end{array}
\right)
\]
and, for $U \in \QQ^{\rect{(s-r) \times s}}$, 
we define the {\it incidence variety} associated to $(A, U)$ as
\begin{equation}
\setV_r(A, U) \coloneqq \big\{ (x, y) \rect{\in \CC^n \times \CC^{s(s-r)}}: A(x) Y(y) = 0, U Y(y) - \Id_{s-r} = 0 \big\}.
\end{equation}

Remark that the matrix $Y(y)$ has full rank $\rect{s-r}$ if and only if there exists
$U \in \QQ^{\rect{(s-r) \times s}}$ of full rank such that $U Y(y) - {\Id_{s-r}} = 0$.
For $A \in (\CC^{\rect{m \times s}})^{n+1}$, $U = (u_{i,j})_{1 \leq i \leq \rect{s-r}, 1 \leq j \leq \rect{s}}
\in \QQ^{(s-r) \times \rect{s}}$, 
\rect{and $c \coloneqq (m+s-r)(s-r)$}, define
\[
\begin{array}{lrcl}
f(A, U) : & \CC^{n+\rect{s(s-r)}} & \to & \CC^{c} \\
                  & (x,y) & \mapsto &  (A(x)  Y(y),\: U  Y(y) - \Id_{s-r})
\end{array}.
\]
Remark that $\setV_r(A, U) = \zeroset{f(A, U)}$ and that the
projection of $\setV_r(A, U)$ over the $x-$space is contained in the
determinantal variety $\setD_r$, by definition.  We denote this projection
map by $\Pi_X$.


\subsection{Data representation}
\label{sec:data}

The input is a \rect{$m \times s$} linear matrix
$A(x) = A_0+x_1A_1+\cdots+x_nA_n$, \rect{with $m \geq s$}, encoded by
the {vector of defining matrices $A=(A_0, A_1, \ldots, A_n)$},
with coefficients in $\QQ$, and an integer $r$ such that $r \leq s-1$.
The vector $A$ is understood as a point in
$(\QQ^{m \times \rect{s}})^{n+1}$.

The output is a finite set sampling the connected components of
$\setD_r \cap \RR^n$.  Indeed, the initial problem is reduced to
isolating the real solutions of an algebraic set $\setZ \subset \CC^n$
of dimension at most $0$, represented by a rational parametrization
$q = (q_0, q_1, \ldots, q_n, q_{n+1}) \in \QQ[t]^{n+2}$, that is with
the representation as in \eqref{ratpar}.



\subsection{Genericity assumptions}
\label{genass}

Our algorithm works under some assumptions on the input $A$ 
(that we denote with the letters $\sfG_1$ and $\sfG_2$).
{These are natural regularity assumptions about the singular
  locus of the locus of low-rank matrices and about smoothness of algebraic
sets, that are often satisfied in practice.}
We will prove below in Section \ref{sec:algo} that these are generic.
We recall that the parameter $r$ is fixed and it holds $0 \leq r < s \leq m$.

\textit{Property $\sfG_1$}. {A $m \times s$ linear matrix $A$
satisfies $\sfG_1$
  if, for all $0\leq p \leq r$, $\setD_p \subset \CC^n$ is either
  empty or \rect{$n-(m-p)(s-p)$}-equidimensional,
  $\sing(\setD_p) = \setD_{p-1}$ and the ideal generated by the
  $(p+1, p+1)$ minors of $A(x)$ is radical.}

{Our algorithm takes as input a linear matrix $A(x)$ assuming
  that $A$ satisfies $\sfG_1$; we will prove that $\sfG_1$ holds
  generically in the sequel. The second property, which we often refer
  to as a regularity property, is defined for any polynomial system}.



\textit{Property $\sfG_2$}. {A polynomial sequence $h = (h_1, \ldots, h_{\rect{k}}) \in
\QQ[x_1, \ldots, x_n]^{\rect{k}}$ satisfies $\sfG_2$ if
\begin{itemize}
\item the ideal $\langle h \rangle$ is a radical ideal of co-dimension $k$, and
\item the algebraic set $\zeroset{h} \subset \CC^n$ is either empty or
  smooth and equidimensional.
\end{itemize}}






\section{Algorithm: description, correctness, complexity} \label{sec:algo}

In this section, we describe the algorithm {\sf LowRank}, prove its
correctness and estimate its {arithmetic} complexity.

\subsection{Formal description}
The input of {\sf LowRank} is a couple $(A,r)$, where $A$ is a tuple
of $n+1$ matrices $A_0, A_1, \ldots, A_n$, of size \rect{$m \times s$}
($m\geq s$), with entries in $\QQ$, and $r \leq s-1$ is an
integer. The algorithm is probabilistic and, upon success, its output
is a rational parametrization encoding a finite set of points
intersecting each connected component of the real algebraic set
$\{x \in \RR^n : \rank \, A(x) \leq r\}$ as described in Section
\ref{sec:data}. 

\subsubsection{Notation} \label{sub:sub:notation}


Recall that given $A \in (\CC^{\rect{m \times s}})^{n+1}$ and $U \in \CC^{(s-r) \times \rect{s}}$,
the polynomial system $f(A, U)$ (of cardinality $c = (m+s-r)(s-r)$) and its zero locus
$\setV_r(A, U)$ have been defined in Section \ref{ssec:smoothness:incvar}.

{\it Change of variables.}
Let $M \in \GL_n(\CC)$. As already explained in Section
\ref{sec:notations}, we denote by $A \circ M$ the affine map
$x \mapsto A(M\,x)$ obtained from $A$ by applying a change of
variables induced by the matrix $M$. In particular
$A = A \circ \Id_n$.  For $M \in \GL_n(\CC)$, and for all
$A \in (\CC^{\rect{m \times s}})^{n+1}, U \in \CC^{(s-r) \times
  \rect{s}}$,
we consequently denote by $f(A \circ M, U)$ the polynomial system
$f(A, U)$ applied to $(M\,x,y)$, and by
$\setV_r(A \circ M, U) = \zeroset{f(A \circ M, U)}$.

{\it Fibers.}
{Given $w \in \CC^n$, we introduce the notation $\pi_w$ for the map
$\pi_w : \CC^n \to \CC$, $\pi_w(x) = w^Tx$, and $\Pi_w : \CC^{n+s(s-r)} \to \CC$,
$\Pi_w(x,y) = w^Tx$, that is $\Pi_w = \pi_w \circ \Pi_X$}. For $w \in \CC^n$ and
$\fiber \in \CC$, we define
\[
\begin{array}{lrcl}
f_{w,\fiber} : & \CC^{n+\rect{s(s-r)}} & \to & \CC^{c+1} \\
           & (x,y) & \mapsto & (f(A,U),\:\Pi_w(x,y)-\fiber)
\end{array}
\]
and denote by
$\setV_{r,w,\fiber}(A, U) = \zeroset{f_{w,\fiber}} \subset \CC^{n+\rect{s}(s-r)}$
the section of $\setV_r$ with the linear space defined by $\Pi_w(x,y)-t=0$.
When parameters are clear from the context, we use the shorter
notation $\setV_{r,w,\fiber}$.

{For $A \in (\CC^{m \times s})^{n+1}$, and
  $w \in (\CC \setminus \{0\})^n$ we denote by
  $\restr{A}{w,t} \in (\CC^{m \times s})^n$ the linear matrix obtained
  by eliminating one variable (up to renaming variable, $x_1$) from
  $A$ using the affine equation $w^t\ x-t=0$.}


{\it Lagrange systems.}  
{Given $w \in \CC^n$, we define}
\[
\begin{array}{lrcl}
\ell(A, U, w) : &  \CC^{n+\rect{s(s-r)+c}} & \to     & \CC^{n+\rect{s(s-r)+c}} \\
                     & (x,y,z)            & \mapsto & (f(A,U), z^T\jac f - (w,0)^T)
\end{array}
\]
where $z=(z_1, \ldots, z_{c})$ is the column vector of Lagrange
multipliers {and $(w,0) \in \CC^{n+s(s-r)}$.
Let $\setZ(A, U, w) = \zeroset{\ell(A, U, w)} \subset \CC^{n+s(s-r)+c}$}.

\subsubsection{Subroutines}

The algorithm {\sf LowRank} uses different subroutines, described as follows.

\noindent
{\sf IsReg}: inputs parameters $A,U$ and {outputs} {true} if
{$A$} satisfies ${\sfG_1}$ and $f(A,U)$ satisfies $\sfG_2$,
{\tt false} otherwise;

\noindent
    {\sf RatPar}: inputs a polynomial system $f$; returns {either the
      empty list if $\zeroset{f}=\emptyset$, or} an error message if $\zeroset{f}$
    is not zero-dimensional, otherwise a rational parametrization of $\zeroset{f}$;

\noindent
{\sf Project}: inputs a rational parametrization of a finite set
$\setZ \subset \CC^N$ and a subset of the variables
$x_1, \ldots, x_N$, and {outputs} a rational parametrization of the
projection of $\setZ$ on the space generated by this subset;

\noindent
{{\sf Lift}: inputs a rational parametrization of a finite set
$\setZ \subset \CC^N$ and a number $\fiber \in \CC$, and {outputs} a
rational parametrization of {the set} $\{(\fiber,x) \in \CC^{N+1} \,
: \, x \in \setZ\}$, {that is the pre-image of $\setZ$ under
the projection on the last $N$ variables};}


\noindent
{\sf Union}: inputs rational parametrizations encoding finite sets
$\setZ_1, \setZ_2$ and {outputs} a rational parametrization of
$\setZ_1 \cup \setZ_2$.

{Moreover, a routine ${\sf nvar}(A)$ outputs the number of variables of
a linear matrix $A$, and a routine ${\sf minors}(A,d)$ outputs the list
of $d \times d$ minors of $A$.}
\subsubsection{The algorithm}




\setlength\fboxrule{1pt}
\setlength\fboxsep{0.3cm}

{This is the formal description of the algorithm. The main routine {\sf LowRank} calls the recursive routine {\sf LowRankRec}; the recursion is on the number of variables of the $m \times s$ linear matrix $A$ (which is always denoted by $n$, and computed by a subroutine ${\sf nvar}$).} \\

\begin{varwidth}{14cm}
  {\bf Algorithm ${\sf LowRank}(A,r)$:}
  \begin{enumerate}
  \item \label{new:main:1} Choose $U \in \QQ^{(s-r) \times \rect{s}}$
  \item \label{new:main:2} If ${\sf IsReg}(A, U) = {\tt false}$ then {return('error: input data not generic')}
  \item \label{new:main:3} return ${\sf LowRankRec}(A,U,r)$
  \end{enumerate}
  {\bf Algorithm ${\sf LowRankRec}(A,U,r)$:}
  \begin{enumerate}
  \item \label{new:rec:4} $n={\sf nvar}(A)$. {If $n \leq (m-r)(s-r)$ then return({\sf RatPar}(${\sf minors}(A,r+1)$))}
  \item \label{new:rec:5} Choose $w \in \QQ^n\setminus \{0\}$, {${\sf P}={\sf Project}({\sf RatPar}(\ell(A, U, w)), {x})$}
  \item \label{new:rec:6} Choose $\fiber \in \QQ$, ${\sf Q}={\sf Lift}({\sf LowRankRec}(\restr{A}{w,t}, U, r), \fiber)$
  \item \label{new:rec:7} return ${\sf Union}({\sf Q}, {\sf P})$.
  \end{enumerate}
\end{varwidth}


{
\begin{example}
  Consider the $4 \times 3$ linear matrix $A = (x_{ij}), i=1 \ldots 4, j=1\dots 3$ with unknown entries. One of the incidence varietes encoding one rank defect in $A$ is given as the zero set of the polynomials $f=(f_1,f_2,f_3,f_4)$ defined by
$$
\left(
\begin{array}{cccc}
f_1 \\
f_2 \\
f_3 \\
f_4
\end{array}
\right)
\coloneqq
\left(
\begin{array}{cccc}
x_{11}y_1 + x_{12}y_2 + x_{13}y_3 \\
x_{21}y_1 + x_{22}y_2 + x_{23}y_3 \\
x_{31}y_1 + x_{32}y_2 + x_{33}y_3 \\
x_{41}y_1 + x_{42}y_2 + x_{43}y_3
\end{array}
\right)
=
\left(
\begin{array}{cccc}
x_{11} & x_{12} & x_{13} \\
x_{21} & x_{22} & x_{23} \\
x_{31} & x_{32} & x_{33} \\
x_{41} & x_{42} & x_{43}
\end{array}
\right)
\left(
\begin{array}{cccc}
y_1 \\
y_2 \\
y_3
\end{array}
\right)
$$
together with an affine equation of the form $u_1y_1 + u_2y_2 + u_3y_3 = 1$, ensuring that a non-zero kernel vector is computed. The Lagrange system $\ell(A, U, w)$ is generated by the critical equations
$$
\begin{array}{cccc}
f_1 = 0, \,\,\,\,\,\, f_2 = 0, \,\,\,\,\,\, f_3 = 0, \,\,\,\,\,\, f_4 = 0, \,\,\,\,\,\, u^Ty - 1 = 0, \,\,\,\,\,\, \sum_{i=1}^4z_i \nabla f_i + z_5 (u,0)^T = (w,0)^T \\
\end{array}
$$
whose solutions are the local minima and maxima of the restriction of a linear function $x \mapsto w^Tx$ to the incidence variety. Once the solutions of this system (that are finitely many, for general choices of $u$ and $w$) are computed (with a rational parametrization) a generic fiber is considered by adding an affine constraint $w^Tx-t=0$, one variable is eliminated and the algorithm is called recursively.
\end{example}
}

\subsection{Correctness}

We start by stating intermediate results which will be used to prove
the correctness of the algorithm.  {The proof of the first
  result below is given in Section \ref{sec:regularity}.}


\begin{proposition}\label{prop:regularity}
  \mynew{Let $m,s,n,r \in \NN$, with $0 \leq r < s \leq m$. The following holds. 
    \begin{enumerate}
    \item There exists a non-empty Zariski open set
      $\zarA \subset (\CC^{m \times \rect{s}})^{n+1}$ such that for
      $A\in \zarA$, $A$ satisfies Property $\sfG_1$.
    \item Let $A$ satisfy $\sfG_1$. There exists a non-empty Zariski open set
      $\zarU_A\subset \CC^{(s-r) \times \rect{s}}$ such that, for
      $U\in \zarU_A$, the following holds:
      \begin{itemize}
      \item[(a)] $f(A, U)$ satisfies Property $\sfG_2$;
      \item[(b)] letting ${\cal N}_r(A, U)$ be the Zariski closure of
        $\setD_r-\Pi_X(\setV_r(A, U))$ and $Z$ an irreducible component of $\setD_p$
        for $0\leq p\leq r$, $Z\cap {\cal N}_r(A, U)$ has co-dimension
        at least $1$ in $Z$.
      \end{itemize}
      \item
        Let $w \in \QQ^n\setminus \{0\}$. There exists a non-empty Zariski open set
        $\zarfiber_A \subset \CC$ such that if $t \in \zarfiber_A \cap \QQ$, then
        $\restr{A}{w,t}$ satisfies Property $\sfG_1$.
   \end{enumerate}}
\end{proposition}

{The second result is proved in Section \ref{sec:dimension}.}

\begin{proposition} \label{prop:dimension} {Let $A$ be in the
    non-empty Zariski open set
    $\zarA\subset (\CC^{m \times \rect{s}})^{n+1}$ and $U$ in the
    non-empty Zariski open set
    $\zarU_A\subset \CC^{(s-r) \times \rect{s}}$ defined in
    Proposition~\ref{prop:regularity}.
    There exists a non-empty Zariski open set
    $\zarW_{A,U} \subset \CC^n$ such that for
    $w\in \zarW_{A,U} \cap \QQ^n$ the following holds.
    
    \begin{enumerate}
    \item $\setZ(A, U, w)$ is finite and $\ell(A, U, w)$ satisfies
      Property $\sfG_2$;
    \item the projection of $\setZ(A, U, w)$ on $(x,y)$ contains the
      set of critical points of the restriction of
      $\Pi_w: (x,y) \to w^Tx$ to $\setV_r$. 
\end{enumerate}}
\end{proposition}


The following proposition will be proved in Section~\ref{sec:closure}. 
\begin{proposition}\label{prop:closure}
  \mynew{Let $\zarA \subset (\CC^{m \times \rect{s}})^{n+1}$ and $\zarU_A\subset \CC^{(s-r) \times \rect{s}}$
    be the non-empty Zariski open sets, and let ${\cal N}_r(A, U) \subset \setD_r$ be the
    Zariski closed set, defined in Proposition~\ref{prop:regularity}.
    Let $A \in \zarA$ and $U \in \zarU_A$.
    Let $\cc\subset \RR^n$ be a connected component
    of $\setD_r\cap \RR^n$.

  There exists a non-empty Zariski open set
  $\zarW'_{A,U} \subset \CC^n$ such that, for
  $w\in \zarW'_{A,U}\cap \QQ^{n}$, the following holds:
  \begin{enumerate}
  \item $\pi_w(\cc)$ is closed 
  \item for $\fiber \in \RR$ in the boundary of $\pi_w(\cc)$, there
    exists $(x,y,z)$ in $\setZ(A , U, w)$ such that $\Pi_w(x,y)=t$
    and $x \notin {\cal N}_r(A, U)$.
\end{enumerate}
}
\end{proposition}



{Observe that in the above statements, the defined non-empty Zariski
open sets (except for the set $\zarA$)
have subscripts indicating which data they depend on. Hence,
starting with $A$ satisfying $\sfG_1$, we highlight the following facts:
\begin{itemize}
\item the non-empty Zariski open set $\zarU_A$ (Proposition~\ref{prop:regularity})
  depends on $A$; 
\item the non-empty Zariski open sets $\zarW_{A,U}$ (Proposition~\ref{prop:dimension})
  and $\zarW'_{A,U}$ (Proposition~\ref{prop:closure}) depend on $A$ and $U$. 
\end{itemize}
}


{\bf Hypothesis $\sfH_1$}.  {In the sequel, $A$ (resp. $U$) is
  assumed to belong to the non-empty Zariski open set $\zarA$
  (resp. $\zarU_A$) defined in Proposition~\ref{prop:regularity}.}


{One also has to ensure that the parameters $w \in \QQ^n$ and $t \in \QQ$ chosen,
respectively, at steps \ref{new:rec:5} and \ref{new:rec:6} belong to the non-empty Zariski
open sets defined in Proposition \ref{prop:dimension} and \ref{prop:closure} at each call
of {\sf LowRankRec}. The choices of random parameters can be stored in an array
\begin{equation} \label{propH}
\left(
(w^{(n)}, t^{(n)}), \ldots, (w^{((m-r)\rect{(s-r)})}, t^{((m-r)\rect{(s-r)})})
\right)
\end{equation}
where the superscript represents the number of variables at the given recursion step
(hence $n$ here is the number of variables of $A$ at the input of {\sf LowRank}).
We also denote by $\zarfiber^{(j)}, \zarW^{(j)}$ and $\zarW^{(j)'}$ the non-empty Zariski
open sets defined by Propositions \ref{prop:regularity}, \ref{prop:dimension} and
\ref{prop:closure}, at the $(n-j+1)-$th recursion call (here we avoid the dependency
on the Zariski open sets).
 
{\bf Hypothesis $\sfH_2$}.
Given $A$ and $U$ satisfying $\sfH_1$, {the parameters} \eqref{propH} satisfy:
\begin{itemize}
\item $w^{(j)} \in \zarW^{(j)} \cap \zarW^{(j)'} \cap \QQ^n\setminus \{0\}$ for $j=(m-r)\rect{(s-r)}, \ldots, n$;
\item $\fiber^{(j)} \in \zarfiber^{(j)} \cap \QQ$ for $j=(m-r)\rect{(s-r)}, \ldots, n$.
\end{itemize}
}

\begin{theorem} \label{theo:correctness}
If $\sfH_1$ and $\sfH_2$ hold, algorithm {\sf LowRank}
returns a rational parametrization whose set of solutions intersects each
connected component of $\setD_r \cap \RR^n$.
\end{theorem}

\begin{proof}
{Suppose first that {$n \leq (m-r)\rect{(s-r)}$}. Since $\sfH_1$ holds, then the
  variety $\setD_r$ is empty {or finite}. Hence the algorithm returns the
  correct output, {that is either the empty list or a rational parametrization
    of the finite set $\setV_r$}.} Thereafter, we proceed by induction on $n$.

Let $n>(m-r)\rect{(s-r)}$ and suppose that for any $(n-1)-$variate linear matrix,
algorithm {\sf LowRank} returns the expected output when $\sfH_1$ and $\sfH_2$ hold,
{namely one point per connected component of $\setD_r \cap \RR^n$}.
Let $A$ be a $n-$variate $\rect{m \times s}$ linear matrix, let $r$ be an integer
such that $0 \leq r \leq s-1$ and let $\cc$ be a connected component of $\setD_r
\cap \RR^n$. Let $U$ be the matrix chosen at Step \ref{new:main:1} of {\sf LowRank}.
{Let $w \in \QQ^n$ be the vector chosen at Step \ref{new:rec:5} of {\sf LowRankRec}
  with input $A,U,$ and $r$. Consider the projection $\pi_w \colon (x_1, \ldots, x_n)
  \to w^Tx$ restricted to $\setV_r(A, U)$.} Since Property $\sfH_1,\sfH_2$ hold, by
Proposition \ref{prop:closure}, {$\pi_w(\cc)$ is closed, and so either $\pi_w(\cc)
  = \RR$ or $\pi_w(\cc) \subsetneq \RR$} is a closed set with non-empty boundary.
We claim that, in both cases, {\sf LowRank} with input $(A,r)$ returns a point which
lies in the connected component $\cc$. This is proved next.

{\it First case.} Suppose first that {$\pi_w(\cc) = \RR$}. In
particular, for $\fiber \in \QQ$ chosen at Step \ref{new:rec:6} of {\sf LowRankRec}
with input $A,U,r$, the set {$\pi_w^{-1}(\fiber)$} intersects
$\cc$, so {$\pi_w^{-1}(\fiber) \cap \cc \neq \emptyset$}. Let {$\restr{A}{w,t}$}
be the $(n-1)-$variate $\rect{m \times s}$ linear matrix obtained from $A$ by
{substituting $x_1 = (1/w_1)(t-\sum_{j=2}^{n}w_jx_j)$ obtained from the linear
  constraint $w^Tx=t$ (indeed, since $w \neq 0$,
one can suppose $w_1 \neq 0$ up to permutation of indices).} Remark that $\pi_w^{-1}(t)
\cap \cc$ is the union of some connected components of the determinantal variety
$\setD^{(n-1)}_r \cap \RR^{n-1} = \{x \in \RR^{n-1} : \rank \, \restr{A}{w,t} \leq r\}$.
Since $\sfH_1$ holds, then $\restr{A}{w,t}$ satisfies ${\sfG_1}$; we deduce
by the induction hypothesis (since $\restr{A}{w,t}$ is $(n-1)-$variate) that
the subroutine {\sf LowRankRec} computes one point in each connected component of
$\setD^{(n-1)}_r \cap \RR^{n-1}$, and so at least one point in $\cc$.

{\it Second case.} Suppose now that $\pi_w(\cc) \neq \RR$. By Proposition \ref{prop:closure},
$\pi_w(\cc)$ is closed. Since $\cc$ is connected, $\pi_w(\cc)$ is a closed interval, and since
$\pi_w(\cc) \neq \RR$ there exists $\fiber$ in the boundary of $\pi_w(\cc)$ such that
$\pi_w(\cc) \subset [\fiber, +\infty)$ or $\pi_w(\cc) \subset (-\infty, \fiber]$. Suppose
without loss of generality that $\pi_w(\cc) \subset [\fiber, +\infty)$, so that $t$ is the
  minimum value attained by $\pi_w$ on $\cc$.

  By Proposition \ref{prop:closure}, there exist
  {$x_2, \ldots, x_n$ in $\RR$, $y \in \CC^{\rect{s(s-r)}}$ and
    $z \in \CC^{c}$} such that, for $x=(\fiber,x_2, \ldots, x_n)$, it
  holds that $(x,y,{z}) \in {\setZ(A,U,w)}$, {and
    $x \not\in{\cal N}_r(A, U)$}.
Then, we conclude that the point $x \in \cc$ appears among the solutions of the
rational parametrization {\sf P} obtained at Step \ref{new:rec:5} of {\sf LowRankRec}.
\end{proof}

{Both cases considered in the above proof are important to be considered. For
instance, the projection of the set defined by the rank defect of the matrix
$\begin{bmatrix}x_1+x_2 & 1\\ 1& x_1-x_2\end{bmatrix}$ is the whole line. The
second case is well illustrated with the set of points at which the matrix
$\begin{bmatrix}x_1 & 1-x_2\\ x_2& x_1\end{bmatrix}$ is rank defective: since it
is compact over $\RR^2$, its projection on a line cannot be surjective.}


\subsection{Complexity analysis}
\label{ssec:complexity-analysis}

In this section we provide an analysis of the complexity of algorithm
{\sf LowRank}. We also give bounds for the maximum number of complex solutions
computed by {\sf LowRank}. We suppose that $A$ satisfies Property $\sfG_1$
and that $f(A,U)$ satisfies Property $\sfG_2$. {Recall that the complexity of
  checking these properties is not evaluated here.}


In order to bound the complexity of {\sf LowRank}, it is essentially
sufficient to bound the complexity of {\sf LowRankRec}. This latter
quantity mainly depends on the subroutine {\sf RatPar} computing the
rational parametrization, whose complexity is computed in Section
\ref{ssec:compl:ratpar}. We rely on routines described in
\cite{SaSc17}, which consists in a symbolic homotopy
algorithm taking advantage of the sparsity structure of the input
polynomial system.


Finally, complexity bounds for the subroutines {\sf Project, Lift, Image}
and {\sf Union} are provided in Section \ref{ssec:compl:subrout}
and refer to results of \cite{PS13}.

\subsubsection{Bounds on the degree of the output of {\sf RatPar}} \label{ssec:compl:mbb}

We consider the subroutine {\sf RatPar} at the first recursion
step of {\sf LowRank}. Its input consists in either the generators $f(A \circ M, U, S)$
of the incidence variety (if $n=(m-r)\rect{(s-r)}$) or the
Lagrange system $\ell(A \circ M, U, S, v)$ (if $n>(m-r)\rect{(s-r)}$).
In both cases, we provide below in Proposition \ref{compl:prop:degree:ratpar}
a bound on the degree of the rational parametrization returned by {\sf RatPar}.

We recall that if $x^{(1)}, \ldots, x^{(p)}$ are $p$ groups of
variables, and $f \in \QQ[x^{(1)}, \ldots, x^{(p)}]$, we say that the
multidegree of $f$ is $(d_1, \ldots, d_p)$ if its degree with respect
to the group of variables $x^{(j)}$ is $d_j$ for $j=1, \ldots, p$.

\begin{proposition} \label{compl:prop:degree:ratpar}
Let $A$ be a $n-$variate $m \times \rect{s}$ linear matrix,
$0 \leq r < s \leq m$ and let $U$ and $w$ be respectively
the parameters chosen at \mynew{step \ref{new:main:1}} of
{\sf LowRank} and at \mynew{step \ref{new:rec:5}} of {\sf LowRankRec}. Suppose that
$\sfH_1$ and $\sfH_2$ hold. Then:
\begin{enumerate}
\item if $n=(m-r)\rect{(s-r)}$, the degree of the output of {\sf RatPar}
  \mynew{at step \ref{new:rec:4}}, is bounded from above by
  $\rect{\binom{m(s-r)}{(m-r)(s-r)}}$;
\item if $n>(m-r)\rect{(s-r)}$, the degree of the output of {\sf RatPar}
  \mynew{at step \ref{new:rec:5}},
  with input $\ell(A,U,w)$, is bounded from above by
  \[
  \delta(m,s,n,r) \coloneqq \sum_{k \in \mathcal{F}_{m,s,n,r}}\binom{\rect{m(s-r)}}{n-k}\binom{n-1}{\rect{k+(m-r)(s-r)-1}}\binom{\rect{r(s-r)}}{k},
  \]
  \rect{with $\mathcal{F}_{m,s,n,r}=\{k : \max\{0,n-m(s-r)\} \leq k \leq \min\{n-(m-r)(s-r),r(s-r)\}\}$.}
\end{enumerate}
\end{proposition}

\begin{proof}[of Assertion 1]
  If $n=(m-r)\rect{(s-r)}$, since $\sfH_1$ holds, the dimension of
  $\setD_r$ is zero. Consequently, the degree of the rational
  parametrization returned by {\sf RatPar} is the degree of
  $\setD_r$. \mynew{We bound this degree by the degree of $\setV_r$,
    which is a finite set by Proposition \ref{prop:regularity}
    (indeed, in the zero-dimensional case, the set
    ${\mathcal N}_r(A,U)$ is empty).}  Since the entries of $f(A, U)$
  have a natural bilinear structure in $\X, \Y$, one takes advantage
  in using the Multilinear B\'ezout bound (see \cite[Chapter 11]{SaSc13}) to bound the degree of the set it
  defines.

From $U Y(\y)-\Id_{s-r}$ one can eliminate $\rect{(s-r)^2}$ variables $\y_{i,j}$.
{Indeed, recall that $U$ has full rank $\rect{s-r}$. Without loss of generality,
we suppose that the last $\rect{s-r}$
columns of $U$ are linearly independent, and hence we eliminate the
variables $y_{i,j}$ corresponding to the last $\rect{s-r}$ rows of $Y(\y)$}.
Abusing notation, we denote by the same symbol $f \subset
\QQ[\x, \y_{1,1}, \ldots, \y_{\rect{r,s-r}}]$ the polynomial system obtained after
this elimination. It is constituted by $m(\rect{s-r})$ polynomials of
multidegree bounded by $(1,1)$ with respect to $\x=(\x_1, \ldots, \x_n)$
and $\y=(\y_{1,1}, \ldots, \y_{r,s-r})$. 

By the Multilinear B\'ezout theorem \cite[Prop. I.1]{SaSc13}, $\deg\zeroset{f}$ is
bounded by the sum of the coefficients of
\[
(s_x+s_y)^{m(\rect{s-r})} \,\,\, \mod \, \left\langle s_x^{n+1}, s_y^{\rect{r(s-r)+1}}\right\rangle \subset \ZZ[s_x,s_y].
\]
Since $n+\rect{r(s-r)}=m(s-r)$, and $(s_x+s_y)^{\rect{m(s-r)}}$ is homogeneous of degree $\rect{m(s-r)}$, the
aforementioned bound equals the coefficient of $s_x^{n}s_y^{\rect{r(s-r)}}$ in the expansion of $(s_x+s_y)^{m(s-r)}$,
that is exactly $\rect{\binom{m(s-r)}{(m-r)(s-r)}}$.
\end{proof}

\begin{proof}[of Assertion 2]
In this case, the input of {\sf RatPar} is the Lagrange system
$\ell(A, U, w)$. Let $f$ be the equivalent system defined
in the proof of Assertion 1. We apply a similar reduction to
$\ell(A, U, w)$. We introduce Lagrange multipliers
$z=[1, z_2, \ldots, z_{\rect{m(s-r)}}]$ (we put $z_1=1$ w.l.o.g.,
since $\ell(A, U, w)$ is defined over the Zariski
open set $z \neq 0$) and we consider polynomials $(g,h)=z^T\jac_1 f$.
Hence the new equivalent system $\ell = (f,g,h)$ is constituted by:
\begin{itemize}
\item $\rect{m(s-r)}$ polynomials of multidegree bounded by $(1,1,0)$;
\item $n-1$ polynomials of multidegree bounded by $(0,1,1)$;
\item $\rect{r(s-r)}$ polynomials of multidegree bounded by $(1,0,1)$.
\end{itemize}
Moreover, by Proposition \ref{prop:dimension}, $\zeroset{f,g,h}$
has dimension at most zero and $(f,g,h)$ satisfies $\sfG_2$.
As above, $\deg \zeroset{f,g,h}$ is bounded by the sum of the coefficients of
\[
(s_x+s_y)^{\rect{m(s-r)}}(s_y+s_z)^{n-1}(s_x+s_z)^{\rect{r(s-r)}} \,\, \mod \, \left\langle
s_x^{n+1}, s_y^{\rect{r(s-r)}+1}, s_z^{\rect{m(s-r)}}\right\rangle \subset \ZZ[s_x,s_y,s_z].
\]
As in the proof of Assertion 1, by homogeneity of the polynomial
and by counting the degrees, the previous sum is given by the
coefficient of the monomial $s_x^{n}s_y^{\rect{r(s-r)}}s_z^{\rect{m(s-r)-1}}$ in
the expansion
\[
\sum_{i=0}^{\rect{m(s-r)}}\sum_{j=0}^{n-1}\sum_{k=0}^{\rect{r(s-r)}}\binom{\rect{m(s-r)}}{i}
\binom{n-1}{j}\binom{\rect{r(s-r)}}{k}s_x^{i+k}s_y^{\rect{m(s-r)}-i+j}s_z^{n-1-j+\rect{r(s-r)}-k}.
\]
The coefficient is obtained by setting the equalities
$i+k=n, \, \rect{m(s-r)}-i+j=\rect{r(s-r)}$ and $n-1-j+\rect{r(s-r)}-k=\rect{m(s-r)}-1$.
These equalities imply $i+k=n=j+k+\rect{(m-r)(s-r)}=j+k+i-j=i+k$ and
consequently one deduces the claimed expression.
\end{proof}


Proposition \ref{compl:prop:degree:ratpar} implies straightforwardly the following
estimate.

\begin{corollary}
Suppose that the hypotheses of Proposition \ref{compl:prop:degree:ratpar}
are satisfied. Then {\sf LowRank} returns a rational parametrization
whose degree is less than or equal to
\[
\rect{\binom{m(s-r)}{(m-r)(s-r)}+\sum_{j=(m-r)(s-r)+1}^{\min\{n,(m+r)(s-r)\}}\delta(m,j,r).}
\]
\end{corollary}
\begin{proof}
Since $\sfH_1$ holds, for $n<(m-r)\rect{(s-r)}$ the algorithm returns the empty
list.
For $m,s,j,r$ let $\mathcal{F}_{m,s,j,r}$ be the set of indices defined in
Proposition \ref{compl:prop:degree:ratpar}. Observe that 
$\mathcal{F}_{m,s,j,r} = \emptyset$ if and only if $j > \rect{(m+r)(s-r)}$.
Hence, the thesis is deduced straightforwardly from bounds given in
Proposition \ref{compl:prop:degree:ratpar}.
\end{proof}

One can also deduce the following bound on $\delta(m, \rect{s}, n, r)$.

\begin{lemma}
For all $m,\rect{s}, n,r$, with $r < s \leq m$, $\delta(m,\rect{s},n,r) \leq \binom{n+m(\rect{s}-r)}{n}^3$.
\end{lemma}
\begin{proof}
This comes straightforwardly from the formula
\[
\binom{a+b}{a}^3 = \sum_{i_1,i_2,i_3=0}^{\min(a,b)}\binom{a}{i_1}\binom{b}{i_1}\binom{a}{i_2}\binom{b}{i_2}\binom{a}{i_3}\binom{b}{i_3}
\]
applied with $a=n$ and $b=m(\rect{s}-r)$, and from the expression of $\delta(m,\rect{s},n,r)$
computed in Proposition \ref{compl:prop:degree:ratpar}.
\end{proof}

\subsubsection{Complexity of {\sf RatPar}} \label{ssec:compl:ratpar}
\new{The computation of the rational parametrization by the subroutine
  {\sf RatPar} is done via the symbolic homotopy algorithm
  \cite{SaSc17}. In this section, we analyze the
  complexity of the algorithm in \cite{SaSc17} for
  our special case.}

We suppose that $n>(m-r)\rect{(s-r)}$ and that the input of {\sf RatPar} is the
equivalent Lagrange system $\ell = \ell(A, U, w) \in
\QQ[x,y,z]^{n-1+\rect{(m+r)(s-r)}}$ built in the proof of Assertion $2$ of Proposition
\ref{compl:prop:degree:ratpar}.
\rect{First, the strategy consists in building} a second polynomial system $\tilde{\ell}
\subset \QQ[x,y,z]$, such that:
\begin{itemize}
\item the length of $\tilde{\ell}$ equals that of $\ell$, that is \rect{$=n-1+(m+r)(s-r)$};
\item for $i=1, \ldots, n-1+\rect{(m+r)(s-r)}$, the support of $\tilde{\ell}_i$ equals that
of $\ell_i$;
\item the solutions of $\tilde{\ell}$ \new{can be computed efficiently (see below)}.
\end{itemize}

Indeed, we remind that by construction, $\ell$ contains three groups
of quadratic polynomials in $\QQ[x,y,z]$, of multidegree respectively bounded
by $(1,1,0),(0,1,1)$ and $(1,0,1)$. We denote by $\Delta_1 \subset \QQ[x,y],
\Delta_2 \subset \QQ[y,z]$ and $\Delta_3 \subset \QQ[x,z]$ the supports of
the three groups, so that for example $\Delta_1=\{1,x_i,y_j,x_iy_j : 1 \leq
i \leq n, 1 \leq j \leq \rect{r(s-r)}\}$, or, equivalently, $\Delta_1$ can be
seen as the subset of $\ZZ^{n+\rect{r(m-r)}}$ made by the exponents of its monomials.
Let $\ell_i$ be with support in $\Delta_1$, $1 \leq i \leq \rect{m(s-r)}$. Hence we
generate two linear forms $g_{i,1} \in \QQ[x]$ and $g_{i,2} \in \QQ[y]$ and we
define $\tilde{\ell}_i(x,y)=g_{i,1}(x)g_{i,2}(y)$. We equivalently generate
polynomials $\tilde{\ell}_j(y,z)=g_{j,1}(y)g_{j,2}(z)$, $\rect{m(s-r)}+1 \leq j \leq
\rect{m(s-r)}+n-1$ and $\tilde{\ell}_k(x,z)=g_{k,1}(x)g_{k,2}(z)$, $\rect{m(s-r)}+n
\leq k \leq n-1+\rect{(m+r)(s-r)}$.

We deduce straightforwardly that $\tilde{\ell}$ satisfies
the above properties. {Indeed, the set $\zeroset{\tilde{\ell}}$ can be computed
by solving systems of linear equations. When the affine polynomials $g_{i,1},g_{i,2},
1 \leq i \leq n-1+\rect{(m+r)(s-r)}$,
are chosen generically, the number of linear systems to be solved equals
the multilinear B\'ezout bound $\delta(m,s,n,r)$, computed in Proposition
\ref{compl:prop:degree:ratpar}. Hence the complexity of solving the starting system
is in ${\bigO}((n+\rect{(m+r)(s-r)})^{\omega}\delta(m,\rect{s},n,r))$, where $\omega$ is the
exponent of linear algebra}. 


In \cite{SaSc17}, the authors build a homotopy path between
$\ell$ and $\tilde{\ell}$, such as
\begin{equation} \label{homotopy}
t \ell + (1-t) \tilde{\ell} \subset \QQ[x,y,z,t]
\end{equation}
where $t$ is a new variable. The system \eqref{homotopy} defines a
$1-$dimensional algebraic set, that is a curve.  We deduce by
\cite[Theorem 1, Corollary 2 and Proposition 5]{SaSc17} that, if the
solutions of $\tilde{\ell}$ are known, one can compute a rational
parametrization of the solution set of system \eqref{homotopy} within
$\softO( (\tilde{n} N\log Q + \tilde{n}^{3}) d d')$ arithmetic
operations over $\QQ$, where:
\begin{itemize}
\item $\tilde{n}$ is the number of variables in $\ell$; 
\item $N= m(\rect{s-r}) \# \Delta_1+ (n-1) \# \Delta_2+r(\rect{s}-r) \# \Delta_3$ ($\#$ is the cardinality);
\item $Q = \max_{i=1,2,3}\{\Vert q \Vert : q \in \Delta_i\}$;
\item $d$ is the number of isolated solutions of $\ell$;
\item $d'$ is the degree of the curve $\zeroset{t \ell + (1-t) \tilde{\ell}}$;
\end{itemize}

\begin{lemma} \label{mult:bounds:homotopy}
Let $\mathcal{F}_{m,\rect{s},n,r}$ and $\delta(m,\rect{s},n,r)$ be the set and the bound defined in
Proposition \ref{compl:prop:degree:ratpar}, and suppose $\mathcal{F}_{m,\rect{s},n,r} \neq \emptyset$.
Then the degree of $\zeroset{t \ell + (1-t) \tilde{\ell}}$ is in
\[
\rect{\bigO\left((n+(m+r)(s-r))\,\min\{n,m(s-r)\}\,\delta(m,s,n,r)\right).}
\]
\end{lemma}
\begin{proof}
We exploit the multilinear structure of $t \ell + (1-t) \tilde{\ell}$.
By the Multilinear B\'ezout theorem,
$\deg \zeroset{t \ell + (1-t) \tilde{\ell}}$ is bounded by the sum of the coefficients of
\[
q(s_x,s_y,s_z,s_t)=(s_x+s_y+s_t)^{m(\rect{s-r})}(s_y+s_z+s_t)^{n-1}(s_x+s_z+s_t)^{r(\rect{s-r})}
\]
modulo $I=\langle s_x^{n+1}, s_y^{r(\rect{s-r})+1}, s_z^{m(\rect{s-r})}, s_t^2\rangle \subset \ZZ[s_x,s_y,s_z,s_t]$.
It is easy to check that $q = q_1 + s_t(q_2+q_3+q_4) + g$ with $s_t^2$ that divides $g$ and
\begin{align*}
q_1 &= (s_x+s_y)^{m(\rect{s-r})}(s_y+s_z)^{n-1}(s_x+s_z)^{r(\rect{s-r})} \\
q_2 &= m(\rect{s-r}) (s_x+s_y)^{m(\rect{s-r})-1}(s_y+s_z)^{n-1}(s_x+s_z)^{r(\rect{s-r})} \\
q_3 &= (n-1)(s_x+s_y)^{m(\rect{s-r})}(s_y+s_z)^{n-2}(s_x+s_z)^{r(\rect{s-r})}\\
q_4 &= r(\rect{s-r})(s_x+s_y)^{m(\rect{s-r})}(s_y+s_z)^{n-1}(s_x+s_z)^{r(\rect{s-r})-1},
\end{align*}
and hence that $q \equiv q_1+s_t(q_2+q_3+q_4) \mod I$. Below, we bound the
contribution of $q_i, i=1 \ldots 4$. The stated bound is given by
the sum of the contributions and follows straightforwardly.

{\it Contributions of $q_1$.}
The contribution of $q_1$ is the sum of its coefficients modulo the ideal
$I'=\langle s_x^{n+1}, s_y^{r(\rect{s-r})+1}, s_z^{m(\rect{s-r})} \rangle$. This has been computed
in Proposition \ref{compl:prop:degree:ratpar}, and coincides with $\delta(m,\rect{s},n,r)$.

{\it The contribution of $q_2$.}
Write $q_2=m(\rect{s-r})\tilde{q}_2$ with $\tilde{q}_2 \in \ZZ[s_x, s_y, s_z]$.
Consequently the contribution is given by the sum of the coefficients of
$\tilde{q}_2$, modulo $I'$, multiplied by $m(\rect{s-r})$. Now, observe that
$\deg \tilde{q}_2 = n-2+\rect{(m+r)(s-r)}$ and that maxima powers admissible
modulo $I'$ are $s_x^{n}, s_y^{r(\rect{s-r})}, s_z^{m(\rect{s-r})-1}$. Hence, three configurations
give a contribution.
\begin{itemize}
\item[(A)] The coefficient of the monomial $s_x^{n-1}s_y^{r(\rect{s-r})}s_z^{m(\rect{s-r})-1}$
  in $\tilde{q}_2$, that is
  \[
  \Sigma_A=\sum_{k=0}^{r(\rect{s-r})}\binom{m(\rect{s-r})-1}{n-1-k}\binom{n-1}{k-1+(m-r)\rect{(s-r)}}\binom{r(\rect{s-r})}{k}.
  \]
\item[(B)] The coefficient of the monomial $s_x^{n}s_y^{r(\rect{s-r})-1}s_z^{m(\rect{s-r})-1}$
  in $\tilde{q}_2$, that is
  \[
  \Sigma_B=\sum_{k=0}^{r(\rect{s-r})}\binom{m(\rect{s-r})-1}{n-k}\binom{n-1}{k-1+(m-r)(\rect{s-r})}\binom{r(\rect{s-r})}{k}.
  \]
\item[(C)] The coefficient of the monomial $s_x^{n}s_y^{r(\rect{s-r})}s_z^{m(\rect{s-r})-2}$
  in $\tilde{q}_2$, that is
  \[
  \Sigma_C=\sum_{k=0}^{r(\rect{s-r})}\binom{m(\rect{s-r})-1}{n-k}\binom{n-1}{k-2+(m-r)(\rect{s-r})}\binom{r(\rect{s-r})}{k}.
  \]
\end{itemize}
So the contribution of $q_2$ equals $m(\rect{s-r})(\Sigma_A+\Sigma_B+\Sigma_C)$.

One easily deduces that $\Sigma_A \leq \delta(m,\rect{s},n,r)$ and $\Sigma_B \leq \delta(m,\rect{s},n,r)$. 
Remember that we suppose $\mathcal{F}_{m,\rect{s},n,r} \neq \emptyset$, that is $\delta(m,\rect{s},n,r) > 0$.
We claim that $\Sigma_C \leq (1+\min\{n,m(\rect{s-r})\})\,\delta(m,\allowbreak \rect{s},n,r)$. Consequently, we conclude that
the contribution of $q_2$ is $m(\rect{s-r})(\Sigma_A+\Sigma_B+\Sigma_C) \in 
\bigO\left(m(\rect{s-r})\,\min\{n,m(\rect{s-r})\}\,\delta(m,\rect{s},n,r)\right)$.

Let us prove this claim. First, denote by
\begin{align*}
\chi_1 &= \max\{0,n-m(\rect{s-r})\} \qquad \,\,\,\,\,\,\,\,\,\, \chi_2 = \min\{r(\rect{s-r}),n-(m-r)(\rect{s-r})\} \\
\alpha_1 &= \max\{0,n+1-m(\rect{s-r})\} \qquad \alpha_2 = \min\{r(\rect{s-r}),n+1-(m-r)(\rect{s-r})\}
\end{align*}
the indices such that $\delta(m,\rect{s},n,r)$ sums over $\chi_1 \leq k \leq \chi_2$
and $\Sigma_C$ over $\alpha_1 \leq k \leq \alpha_2$. Remark that
$\chi_1 \leq \alpha_1$ and $\chi_2 \leq \alpha_2$. Finally, denote by
$\varphi(k)$ the $k-$th term in the sum defining $\Sigma_C$, and
by $\gamma(k)$ the $k-$th term in the sum defining $\delta(m,\rect{s},n,r)$.

For all indices $k$ admissible for both $\delta(m,\rect{s},n,r)$ and $\Sigma_C$,
that is for $\alpha_1 \leq k \leq \chi_2$, one gets, by basic properties
of binomial coefficients (we apply $\binom{a}{b-1}=\frac{b}{a-b-1}\binom{a}{b}$), that
\[
\varphi(k) = \Psi(k)\,\gamma(k) \qquad \text{with} \,\,\, \Psi(k) = \frac{k-1+(m-r)(\rect{s-r})}{n-k-(m-r)(\rect{s-r})-1}.
\]
When $k$ runs over all admissible indices, the rational function $\Psi(k)$
is non-decreasing monotone, and its maximum is attained in $\Psi(\chi_2)$
and is bounded by $\min\{n,m(\rect{s-r})\}$. Three possible cases can hold:
\begin{enumerate}
\item $\alpha_1=0$. Hence $\chi_1=0$, $\alpha_2=r(\rect{s-r})$ and $\chi_2=r(\rect{s-r})$.
  We deduce straightforwardly from the above discussion that $\Sigma_C \leq \min\{n,m(\rect{s-r})\} \delta(m,\rect{s},n,r)$;
\item $\alpha_1=n-m(\rect{s-r})+1$ and $\chi_1=n-m(\rect{s-r})$. We deduce that $\chi_2=\alpha_2=r(\rect{s-r})$
  and that
  $
  \Sigma_C = \sum_{k=\alpha_1}^{\chi_2}\varphi(k) \leq \varphi(\alpha_1)+\min\{n,m(\rect{s-r})\}\,\delta(m,\rect{s},n,r) \leq
  (1+\min\{n,m(\rect{s-r})\})\,\delta(m,\rect{s},n,r);
  $
\item $\chi_1=0$ and $\alpha_1=n-m(\rect{s-r})+1$. Hence, we deduce the chain of inequalities
  $0 \leq n-m(\rect{s-r})+1 \leq 1$. Hence, either this case coincides with case $2$
  (if $n=m(\rect{s-r})$) or we deduce that $n=m(\rect{s-r})-1$, and we fall into case 1.
\end{enumerate}

{\it The contribution of $q_3$ and $q_4$.}
Following exactly the same path as in the case of $q_2$, one respectively deduces
that the contribution of $q_3$ is in $\bigO \left(n\,\min\{n,m(\rect{s-r})\}\,\delta(m,\rect{s},n,r)\right)$
and that of $q_4$ is in $\bigO\left(r(\rect{s-r})\,\min\{n,m(\rect{s-r})\}\,\delta(m,\rect{s},n,r)\right)$.
\end{proof}


We can now state the main result of this paragraph. {We denote by
  $\softO(D) = \bigO(D\,\log^a(D))$, that is linear complexity up to logarithmic
factors.}

\begin{theorem} \label{theo:compl:ratpar}
Let $n>(m-r)\rect{(s-r)}$. Let $A$ be a $n-$variate $m \times \rect{s}$ linear matrix,
$0 \leq r < \rect{s \leq m}$ and let $U$ be the matrix
chosen in step \ref{new:main:1} of {\sf LowRank}. Let $\delta=\delta(m,\rect{s},n,r)$ be
the bound defined in Proposition \ref{compl:prop:degree:ratpar}. Then, {\sf RatPar}
returns a rational parametrization within
\[
\rect{\softO\left((n+(m+r)(s-r))^{6}\,\delta^2\right)}
\]
arithmetic operations.
\end{theorem}
\begin{proof}
Following the notation introduced above, $\tilde{n} = n-1+\rect{(m+r)(s-r)}$.
the bound for $d$ is $\delta$
and is given in Proposition \ref{compl:prop:degree:ratpar} and
a bound for $d'$ is given in Lemma \ref{mult:bounds:homotopy}, and
is in $\softO(\tilde{n}^2 \delta)$.
Moreover, $N \in {O}(mnr(\rect{s-r})^2)$, and hence $N \in {O}(\tilde{n}^3)$.
The proof follows from \cite[Proposition 5]{SaSc17},
since the maximum diameter of $\Delta_1,\Delta_2,\Delta_3$ is bounded
above by $\tilde{n}$, that is $Q \leq \tilde{n}$ \rect{in the notation above}.
\end{proof}

\subsubsection{Complexity of subroutines}
\label{ssec:compl:subrout}
For these complexity bounds, we refer to those given in \cite[Lemma 3
and 4]{PS13} (see
\cite[Lemma J.3, J.5 and J.6]{SaSc13} for a unified treatment of these algorithms)
from which they are obtained straightforwardly.

\begin{proposition} \label{compl:prop:allthesubrout}
Let $\delta(m,\rect{s},n,r)$ be the bound defined in Proposition \ref{compl:prop:degree:ratpar}.
At the first recursion step of {\sf LowRankRec}, the following holds:
\begin{itemize}
\item the complexity of {\sf Project} is in $\softO\left((n+\rect{(m+r)(s-r)})^2 \: (\delta(m,\rect{s},n,r))^2\right)$;
\item the complexity of {\sf Lift} is in $\softO\left((n+\rect{(m+r)(s-r)}) \: (\delta(m,\rect{s},n,r))^2\right)$;
\item the complexity of {\sf Union} is in $\softO\left((n+\rect{(m+r)(s-r)}) \: (\delta(m,\rect{s},n,r))^2\right)$.
\end{itemize}
\end{proposition}





\section{Regularity of the incidence variety} \label{sec:regularity}

The goal of this section is to prove Proposition \ref{prop:regularity}.
We introduce the notation $\mathfrak{B}$ representing a $m\times s$ matrix,
with $m\geq s$, whose entries are indeterminates $\mathfrak{b}=(\mathfrak{b}_{i,j})$; similarly,
the $(s-r)\times s$ matrix $\mathfrak{U}$ whose entries are indeterminates
$\mathfrak{u}=(\mathfrak{u}_{i,j})$, and we use the notation $\mathfrak{a}=(\mathfrak{a}_{\ell,i,j})$
to denote generic entries of the linear matrix $A(x)$. All the projection maps
will be denoted by $\pi$ when source and target spaces are clear from the context.

\mynew{
  
\begin{proof}[of Assertion 1 of Proposition \ref{prop:regularity}]
    By \cite[Prop 3.1]{ranestad2012algebraic}, there exists a
    non-empty Zariski open set
    $\zarA'_1\subset (\CC^{m \times \rect{s}})^{n+1}$ such
    that for $A\in \zarA'_1$ and all $0\leq p \leq r$,
    $\setD_p \subset \CC^n$ is either empty or
    \rect{$n-(m-p)(s-p)$}-equidimensional, and
    $\sing(\setD_p) = \setD_{p-1}$.
    It remains to prove that there exists a non-empty Zariski open set
    $\zarA_1'' \subset (\CC^{m \times \rect{s}})^{n+1}$ such
    that, for $A\in \zarA_1''$ and all $0\leq p \leq r$, the ideal
    generated by the $(p+1, p+1)$ minors of $A(x)$ is
    radical. Defining $\zarA_1$ as the intersection of $\zarA_1'$ and
    $\zarA_1''$ leads to the following conclusion: there exists a
    non-empty Zariski open set
    $\zarA_1\subset (\CC^{m \times \rect{s}})^{n+1}$ such that for
    $A\in \zarA_1$, $A$ satisfies $\sfG_1$.

   {}
    
   By \cite[Prop. 12.2]{harris1992algebraic}, the ideal
   $I \subset \CC[\mathfrak{b}]$ generated by the $(p+1, p+1)$ minors
   of $\mathfrak{B}$ is radical and $\zeroset{I}$ is prime, of
   co-dimension $(m-p)(s-p)$. We deduce that the Jacobian matrix of
   the set of generators of $I$ has full rank $(m-p)(s-p)$ when
   instantiated at a smooth point (having rank exactly $p$) of
   $\zeroset{I}$. A simple dimension count shows that one can apply
   Bertini's theorem \cite[Theorem 17.16]{harris1992algebraic} when
   adding generic linear forms
   $\{L_{i,j} = \mathfrak{b}_{i,j}-\sum_{\ell} \mathfrak{a}_{\ell,i,j}
   x_\ell\}$
   (where we put by convention $x_0=1$) to the ideal
   $I'=I+\langle L_{i,j} \rangle$ to deduce that one obtains a prime
   ideal. Using the Jacobian crierion \cite[Theorem
   16.19]{Eisenbud95}, we deduce that the rank of the Jacobian matrix
   of $I'$ equals the rank of the Jacobian matrix of $I$ (which is
   $(m-p)(s-p)$) plus $ms = \#\{L_{i,j}\}$, since an identity
   submatrix will appear in correspondence with the derivatives with
   respect to $\mathfrak{a}_{0,i,j}$. We consider now the restriction
   to $\zeroset{I'}$ of the projection
   $\pi(\mathfrak{b},\mathfrak{a},x)= \mathfrak{a}$ eliminating
   variables $\mathfrak{b},x$. By Sard's theorem the singular values of
   $\pi$ lie in a Zariski-closed set of the image of $\pi$.  We deduce
   that there exists a non-empty and Zariski-open set
   $\zarA_1'' \subset (\CC^{m \times \rect{s}})^{n+1}$ such that if
   $A \in \zarA_1''$, the ideal
   $I'' = I+\langle \mathfrak{b}_{i,j}-\sum_\ell a_{\ell,i,j}x_\ell
   \rangle \subset \CC[\mathfrak{b},x]$
   is radical. Thus the intersection $I'' \cap \CC[x]$ (that
   eliminates the variables $\mathfrak{b}$) still yields a radical
   ideal in $\CC[x]$.  Finally, note that this elimination ideal
   coincides with the ideal generated by the $(p+1,p+1)$ minors of
   $A(x)$ (indeed, the elimination is performed by substituting the
   generic entries of $\mathfrak{B}$ with the entries of $A$). We
   conclude that if $A \in \zarA_1''$, the ideal of $(p+1,p+1)$ minors
   of $A(x)$ is radical.
\end{proof}

\begin{proof}[of Assertion 2 of Proposition \ref{prop:regularity}]
    We prove now that there exists a non-empty Zariski open set
    $\zarA_2\subset (\CC^{m \times \rect{s}})^{n+1}$ such that the
    following holds. For $A\in \zarA_2$, there exists a non-empty
    Zariski open set $\zarU_A\subset \CC^{(s-r) \times \rect{s}}$ such
    that for $U\in \zarU_A$, Assertions $(2a)$ and $(2b)$ of
    Proposition~\ref{prop:regularity} are satisfied. Taking $\zarA$ as
    the intersection of $\zarA_1$ and $\zarA_2$ will end the proof.

    Let
    $\mathfrak{F} \subset
    \CC[\mathfrak{a},\mathfrak{b},\mathfrak{u},y]$
    denote the vector of polynomials consisting of: the generic linear
    forms
    $\{L_{i,j} = \mathfrak{b}_{i,j}-\sum_{\ell}
    \mathfrak{a}_{\ell,i,j} x_\ell\}$
    as in the proof of Assertion 1, and the entries of
    $\mathfrak{B} Y$ and $\mathfrak{U}Y-\Id_{s-r}$. First, remark that
    all solutions $(A,B,U,y)$ of $\mathfrak{F}=0$ satisfy
    $\rank\,U = \rank\,Y = s-r$ by the classical condition
    $\rank \, MN \leq \min \{\rank\,M, \rank\,N\}$.  The Jacobian
    matrix $\jac \mathfrak{F}$ of $\mathfrak{F}$ has full rank when
    restricted to $\zeroset{F}$, since we can construct a non-singular
    block sub-matrix of $\jac \mathfrak{F}$ made by the following
    blocks:
    \begin{itemize}
    \item
      the derivatives of forms $L_{i,j}$ with respect to $\mathfrak{a}_{0,i,j}$, a $ms \times ms$
      block equal to $\Id_{ms}$;
    \item
      the derivatives of $\mathfrak{B} Y$ with respect to $\mathfrak{b}$,
      a $m(s-r) \times ms$ block-diagonal matrix with $m$ blocks equal to $Y^T$;
    \item
      the derivatives of $\mathfrak{U}Y-\Id_{s-r}$ with respect to $\mathfrak{u}$,
      a $(s-r)^2 \times s(s-r)$ block-diagonal matrix with $s-r$ blocks equal to $Y^T$.
    \end{itemize}

    By the Jacobian criterion, $\mathfrak{F}$ satisfies $\sfG_2$,
    hence $\zeroset{\mathfrak{F}}$ is smooth and equidimensional. We
    consider the restriction of
    $\pi(\mathfrak{a},\mathfrak{b}, \mathfrak{u},x,y) = \mathfrak{a}$
    to $\zeroset{\mathfrak{F}}$. Applying Sard's theorem, we obtain a
    non-empty and Zariski-open set $\zarA_2$ such that, for
    $A \in \zarA_2$, the ideal generated by $\mathfrak{F}'$ (obtained
    from $\mathfrak{F}$ by instantiating $\mathfrak{a}$ to the entries
    of $A$) satisfies $\sfG_2$.

    Let us fix $A \in \zarA_2$. Considering the new projection
    $\pi: (\mathfrak{b},\mathfrak{u},x,y) \to \mathfrak{u}$ restricted
    to $\zeroset{\mathfrak{F}'}$, and applying Sard's theorem implies
    that there exists $\zarU'_A \subset \CC^{(s-r) \times s}$,
    non-empty and Zariski-open, such that if $U \in \zarU'_A$,
    instantiating $\mathfrak{u}$ to $U$ yields a radical ideal
    $I \subset \CC[\mathfrak{b},x,y]$, with the Jacobian matrix of $I$
    full rank at every solution. Now, $f(A,U)$ with $A \in \zarA_2$
    and $U \in \zarU'_A$ generates the elimination ideal
    $I \cap \CC[x,y]$, hence it is still radical. Since $f(A,U)$ is
    obtained from $I$ by instantiating $\mathfrak{b}$ to the entries
    of $A$, the Jacobian matrix $\jac{f(A,U)}$ is a submatrix of the
    Jacobian matrix of a set of generators of $I$, which has full
    rank. Since the polynomials in $f(A,U)$ do not depend on variables
    $\mathfrak{b}$, it is easily seen that $\jac{f(A,U)}$ has full
    rank too. Hence we deduce that, for $A \in \zarA_2$ and
    $U \in \zarU'_A$, ${f(A,U)}$ satisfies $\sfG_2$, as claimed.

    It remains to prove that there exists a non-empty Zariski open set
    $\zarU''_A\subset \CC^{(s-r) \times s}$ such that for
    $U\in \zarU''_A$, assertion $(2b)$ holds. Finally taking the
    intersection of $\zarU'_A$ and $\zarU''_A$ to define $\zarU_A$
    ends the proof.  Let $A \in \zarA_2$, $p \leq r$ and let $Z$ be
    one of the (finitely-many) irreducible components of $\setD_p$,
    and let $d$ be its dimension (all such components have the same
    dimension, since we proved that $\setD_p$ is empty or
    equidimensional). Intersecting $Z$ with $d$ general hyperplanes,
    we get a finite number of smooth points in $Z$. For every such
    point $x \in Z$, we first build a Zariski-open set
    $\zarU''_{A,p,Z,x}\subset \CC^{(s-r) \times s}$, as follows.

    The rank of $A(x)$ is $p$ since $x$ is a smooth point of $Z$
    (because $Z$ is an irreducible component of the Zariski closure of
    the set of points at which $A$ has rank $p$). The polynomial
    system $y \mapsto f(A,U)$ is linear in $y$. Since
    $\rank A(x) = p$, the condition $A(x)Y(y)=0$ defines a linear
    space $V = \{Y(y) \in \CC^{\rect{s \times (s-r)}} : A(x)Y(y)=0 \}$
    of dimension ${(s-p)(s-r)}$. Since $p \leq r$, remark that
    ${(s-r)^2 \leq (s-p)(s-r)}$. For a generic
    $U \in \CC^{(s-r) \times s}$, the ${(s-r)^2}$ affine equations
    $UY(y)-\Id_{s-r}=0$ define a linear space intersecting $V$.  Hence
    there exists a non-empty Zariski open set
    $\zarU''_{A,p,Z,x} \subset \CC^{(s-r) \times s}$ such that, if
    $U \in \zarU''_{A,p,Z,x}$, the linear system
    $A(x)Y(y)=0, UY(y)-\Id_{s-r}=0$ has at least one solution.

    One concludes by defining
    $$
    \zarU''_A = \bigcap_{p \leq r} \bigcap_{Z \subset \setD_p \cap \RR^n} \bigcap_{x \in Z} \zarU''_{A,p,Z,x}
    $$
    which is non-empty and Zariski open by the finiteness of the number of irreducible components of
    $\setD_p \cap \RR^n$ and of the set of points $x$ in $Z$.
\end{proof}

\begin{proof}[of Assertion 3 of Proposition \ref{prop:regularity}]
  By Sard's Theorem, the critical values of the projection
  $\pi_w(x) = w^Tx$ are finitely many, hence the regular values of
  this map define a non-empty Zariski open set
  $\zarfiber_A \subset \CC$.  For $w$ as in the hypothesis, and
  $t \in \zarfiber_A$, we denote by
  $\setD_p'= \{x \in \CC^{n-1} : \rank\,\restr{A}{w,t}(x) \leq p \}$.
  As in the proof of Assertion 1, since $A \in \zarA$, we deduce that
  if $t \in \zarfiber_A$, then the ideal
  $I=\left\langle \{(p+1)\times(p+1) \text{ minors of } A\},
    \pi_w(x)-t \right\rangle$
  is still radical and $\zeroset{I}$ has co-dimension
  $(m-p)(s-p)+1$. Remark that $I \cap \RR[x_2,\ldots,x_n]$ is
  generated by the $(p+1)\times(p+1)$ minors of $\restr{A}{w,t}$, it
  is still radical and $\setD_p'=\zeroset{I \cap \RR[x_2,\ldots,x_n]}$
  has co-dimension $(m-p)(s-p)$ in $\CC^{n-1}$. Moreover, always for
  $t \in \zarfiber_A$, a point in $\setD_p$ is regular if and only if
  it is regular in $\setD_p \cap \zeroset{\pi_w(x)-t}$.  Hence we
  deduce that for $t \in \zarfiber_A$, the matrix $\restr{A}{w,t}$
  satisfies Property $\sfG_1$, as claimed.
\end{proof}

}

\section{Dimension of Lagrange systems} \label{sec:dimension}

The goal of this Section is to prove Proposition \ref{prop:dimension}.
We first need to give a local description of the incidence variety $\setV_r$
and of the solution set $\setZ(A,U,w)$ of the Lagrange system $\ell(A,U,w)$.

\subsection{Local description of the incidence variety} \label{sectionlocal}
\mynew{As in our previous work \cite{HNS2014}, we need to compute equations for the incidence
sets lifting the determinantal varieties. Here we generalize the equations in \cite[Section 4]{HNS2014}
to the case of low-rank rectangular matrices.}

Let $A=A_0+x_1A_1+\ldots+x_nA_n$ be a $n-$variate $m \times \rect{s}$ linear
matrix with coefficients in $\QQ$, and let $r \leq \rect{s-1}$.
From now on, for $g \in \QQ[x]$, we denote by $\QQ[x]_g$ the localization of
the ring $\QQ[x]$ at $\langle g \rangle$, \mynew{see \cite{Eisenbud95}}.
We recall that the polynomial system defining $\setV_r$ is given
by $f(A, U)$, which contains the entries of $A(x) Y(y)$
and $U Y(y)-\Id_{s-r}$. For $p \leq r$, let $N$ be the upper-left
$p \times p$ submatrix of $A$, so that
\begin{equation} \label{blockdivision}  
A =
\left(
\begin{array}{cc}
N & Q \\
P^T & R
\end{array}
\right)
\end{equation}
with $P \in \QQ[x]^{p \times (m-p)},\rect{Q \in \QQ[x]^{p \times (s-p)}}$ and $R \in \QQ[x]^{(m-p)  \times \rect{(s-p)}}$.
The next Lemma computes the equations of $\setV_r$ in the local ring $\QQ[x,y]_{\det N}$.

\begin{lemma} \label{lemma:local:incidence} Let $A,N,Q,P,R$ be as
  above, and $U$ be any full-rank matrix. Then there exist
  $\{q_{i,j}\}_{1 \leq i \leq p, 1 \leq j \leq
    \rect{s-p}},\{q'_{i,j}\}_{\rect{1 \leq i \leq m-p,1 \leq j \leq
      s-p}} \subset \QQ[x]_{\det N}$
  such that the constructible set
  $\setV_r \cap \{(x,y) : \det N(x) \neq 0\}$ is defined by the
  equations
\begin{align*}
\rect{y_{i,j} - q_{i,1}y_{p+1,j} - \ldots - q_{i,s-p}y_{s,j} = 0} & \qquad \rect{1 \leq i \leq p, 1 \leq j \leq s-r} \\
\rect{q'_{i,1}y_{p+1,j} + \ldots + q'_{i,s-p}y_{s,j} = 0} & \qquad \rect{1 \leq i \leq m-p, 1 \leq j \leq s-r} \\
U Y(y) - \Id_{s-r} = 0&.
\end{align*}
\end{lemma}
\begin{proof}
We denote by $Y^{(1)}$ and $Y^{(2)}$ the submatrices of $Y(y)$ containing respectively the first
$p$ rows and the last $\rect{s-p}$ rows. We also use the block-division of $A$ as in
\eqref{blockdivision}. We claim that in $\QQ[x,y]_{\det N}$ the $m(\rect{s-r})$ equations
$A(x) Y(y) = 0$ are equivalent to the $m(\rect{s-r})$ equations:
\[
\left(
\begin{array}{c}
\Id_p Y^{(1)} + N^{-1} Q Y^{(2)} \\
\Sigma(N) Y^{(2)}
\end{array}
\right) = 0
\]
where $\Sigma(N) = R-P^T N^{-1}Q$
is the Schur complement of $N$ in $A$. Renaming the entries of
$N^{-1}Q$ and $\Sigma(N)$ concludes the proof. To prove the claim, remark that since $\det N \neq 0$,
$A(x) Y(y) = 0$ if and only if
\[
\left(\begin{array}{cc} \Id_p & 0 \\ -P^T & \Id_{m-p} \end{array} \right)
\left( \begin{array}{cc} N^{-1} & 0 \\ 0 & \Id_{m-p} \end{array} \right)
\left( \begin{array}{cc} N & Q \\ P^T & R \end{array} \right) Y(y) = 0.
\]
\end{proof}

\subsection{The rank at a critical point} \label{subsec:rank}
Given $A, N, P, Q, R, \Sigma(N)$ as above, let 
$$
\widetilde{A} =
\left(
\begin{array}{cc}
\Id_p & N^{-1}Q \\
0 & \Sigma(N)
\end{array}
\right).
$$
Lemma \ref{lemma:local:incidence} implies that the equations of
$\setV_r$ in the open set $\{(x,y) : \det N \neq 0\}$ can be rewritten
as $\widetilde{A}(x) Y(y)=0$ and $U Y(y) - \Id_{s-r}=0$: the polynomial
entries of the above expressions are elements of the local ring
$\QQ[x]_{\det \, N}$. Now, from the first group of relations
$\widetilde{A}(x) Y(y)=0$ one eliminates variables
$\{y_{i,j}\}_{1 \leq i \leq p, 1 \leq j \leq \rect{s-r}}$, which can
be expressed as polynomial functions of $x$ and
$\{y_{i,j}\}_{p+1 \leq i \leq \rect{s}, 1 \leq j \leq \rect{s-r}}$.
That is, using the notations introduced in Lemma
\ref{lemma:local:incidence}, we can express the entries of $Y^{(1)}$
as polynomials in $x$ and in the entries of $Y^{(2)}$.

Now, consider the relations $U Y(y) - \Id_{s-r} = 0$ where the entries
of $Y^{(1)}$ have been eliminated. This is a linear system in the
entries of $Y^{(2)}$ with coefficients in $\QQ[x]_{\det N}$. Since
$\Id_{s-r}$ is full-rank, then $U$ is full-rank and hence
$U Y(y) - \Id_{s-r} = 0$ consists of $(\rect{s-r})^2$ independent
relations.  Finally one can eliminate $(\rect{s-r})^2$ among the
$\rect{(s-p)(s-r)}$ entries of $Y^{(2)}$ (suppose the first
$(\rect{s-r})$ rows) and re-write $\Sigma(N)Y^{(2)} = 0$ as
$(m-p)(\rect{s-r})$ relations in $x$ and in the last
$(r-p)(\rect{s-r})$ entries of $Y^{(2)}$.

Let us call $F$ this polynomial system, and consider a vector of Lagrange multipliers
$z = (z_1, \ldots, z_{(m-p)(\rect{s-r})})$ and the polynomial system
\[
(g_1, \ldots, g_n) = z^T D_xF - (w_1, \ldots, w_n).
\]
The solutions to the above polynomial system contain the critical points
of \new{the projection $\pi_w \colon \CC^n \to \CC$, $\pi_w(x) = w^T x \coloneqq w_1x_1+\cdots+w_nx_n$,}
restricted to $\setV_r \cap \{(x,y) : \det N \neq 0\}$.
The next Lemma shows that, when $w$ is generic in $\CC^n$, the solutions to
the Lagrange systems project on points of $\setD_r$ with rank exactly $r$
(namely in $\setD_r \setminus \setD_{r-1}$).


\begin{lemma} \label{lemma:crit:points} Let $A,U$ be as above and
  suppose that $A$ satisfies $\sfG_1$.  Let $p \leq r-1$ and let
  $g=(g_1, \ldots, g_n)$ be the polynomial system defined above.
  Then there exists a non-empty Zariski open set
  $\widetilde{\mathscr{W}}_{A,U} \subset \CC^{n}$ such that the following holds:
  if $w \in \widetilde{\mathscr{W}}_{A,U}$ then the system $g(x,y,z,w) = 0$ has
  no solutions in $x,y,z$.
\end{lemma}
\begin{proof}
Let $C \subset \CC^{2n+\rect{(r-p)(s-r)+(m-p)(s-r)}}$ be the
constructible set defined by $g=0$ and by $\det N \neq 0$ and
$\rank \, A(x) = p$, \rect{where the coordinates of $w$ are variables}.
Let $\overline{C}$ be the Zariski closure of $C$.
Let $\pi_x \colon (x,y,z,w) \to x$ be the restriction of the projection on the
first $n$ variables to $C$.
The image $\pi_x(C)$ is dense in $\setD_p$. Hence, since $A$ satisfies $\sfG_1$, it has dimension at most
$\rect{n-(m-p)(s-p)}$. The fiber of $\pi_x$ over a generic point
$x \in \setD_p$ is the graph of the polynomial function
$w = z^TD_xF$, and so it has codimension $n$ and dimension
$\rect{(r-p)(s-r)+(m-p)(s-r) = (m+r-2p)(s-r)}$. By the Theorem of
the Dimension of Fibers \cite[Sect. 6.3, Th. 7]{Shafarevich77}
one deduces that the dimension of $C$ (and of $\overline{C}$) is
at most
\[\rect{n-(m-p)(s-p)+(m+r-2p)(s-r) = n-(r-p)(m-p+r-s)\allowbreak{}<n}.\]
\allowbreak{} We deduce that the projection of $\overline{C}$ onto the space $\CC^n$ of $w$ is a
constructible set of dimension at most $n-1$, and it is included
in a hypersurface $H \subset \CC^n$. Defining $\widetilde{\mathscr{W}}_{A,U} = \CC^n
\setminus H$ ends the proof.
\end{proof}

\subsection{Local description of the Lagrange system} \label{subsec:localLagr}

We consider the incidence variety $\setV_r=\setV_r(A, U)$ and the restriction of
the projection $\Pi_w(x,y)=w^T x$ to $\setV_r$, with $w \in \CC^n$. Under the hypothesis
that $A$ satisfies \mynew{$\sfG_1$ and that $f(A,U)$ satisfies $\sfG_2$}, the set $\setV_r$
is either empty or smooth and equidimensional of codimension $c \coloneqq (m+s-r)(s-r)$.
The set of critical points of the restriction of $\Pi_w$ to $\setV_r$ is the projection
on the $(x,y)$-space of the solutions of the Lagrange system $\ell(A,U,w)$:
\begin{equation}
\label{eqlag}
f(A, U)=0
\qquad
(g,h) \coloneqq
z^T
\left(
\begin{array}{cc}
D_xf & D_yf \\
w^T & 0
\end{array}
\right)=0,
\end{equation}
where $z=(z_1, \ldots, z_{c},1)$. 
The polynomial system \eqref{eqlag} consists of $n+c+\rect{s(s-r)}$ polynomials in
$n+c+\rect{s(s-r)}$ variables.
We show that it can be re-written in a local form when we consider the local description
of the incidence variety $\setV_r$ as in Section \ref{sectionlocal}.

We use the block-division of matrix $A$ as in \eqref{blockdivision} with $p=r$
and without loss of generality one can assume to work in the open set $\det \,
N \neq 0$, with $N$ the upper-left $r \times r$ submatrix of $A$. 
We deduce by Lemma \ref{lemma:local:incidence} that the local equations of $\setV_r$
are
\[
Y^{(1)}=-N^{-1}QY^{(2)}, \qquad \Sigma(N) Y^{(2)}=0, \qquad
U^{(1)}Y^{(1)}+U^{(2)}Y^{(2)}=\Id_{s-r},
\]
where $Y^{(1)},Y^{(2)}$ is the row-subdivision of the matrix $Y(y)$ as
in Lemma \ref{lemma:local:incidence} and $U^{(1)},U^{(2)}$ is the
corresponding column-subdivision of $U$. From the first and third
groups of equations one obtains that
$\Id_{s-r}=U^{(1)}(-N^{-1}QY^{(2)})+U^{(2)}Y^{(2)}=(-U^{(1)}N^{-1}Q+U^{(2)})Y^{(2)}$.
Since $\Id_{s-r}$ is full-rank, then $Y^{(2)}$ and
$-U^{(1)}N^{-1}Q+U^{(2)}$ are non-singular, and so:
\begin{itemize}
\item the second group of equations can be re-written as $\Sigma(N) = 0$;
\item the third group of equations can be re-written as
  $Y^{(2)} = (-U^{(1)}N^{-1}Q+U^{(2)})^{-1}$.
\end{itemize}
The entries of $\Sigma(N)$ in the local ring $\QQ[x]_{\det N}$ are exactly
the $\rect{(m-r)(s-r)}$ minors of $A(x)$ obtained as determinants of the $(r+1) \times
(r+1)$ submatrices of $A(x)$ containing $N$
(see for example the proof of \cite[Proposition 3.2.7]{SaSc13}). Since $A$
satisfies $\sfG_1$, the Jacobian
$D_x[\Sigma(N)]_{i,j}$ of the vector of entries of $\Sigma(N)$ has full-rank
at each point $x$ such that $\rank \, A(x) = r$.

We call $f'=(f'_1, \ldots, f'_c)$ the local equations represented by the
entries of $\Sigma(N)$, $Y^{(1)}+N^{-1}QY^{(2)}$ and $Y^{(2)} - (-U^{(1)}N^{-1}Q+U^{(2)})^{-1}$.
The Jacobian matrix of $f'$ has the form
\[
\jac f' =
(D_xf' \,\,\, D_yf') =
\left(
\begin{array}{cc}
D_x[\Sigma(N)]_{i,j} & 0_{\rect{(m-r)(s-r) \times s(s-r)}} \\
\star & 
\begin{array}{cc}
\Id_{\rect{r(s-r)}} & \star \\
0 & \Id_{\rect{(m-r)(s-r)}}
\end{array}
\end{array}
\right)
\]
We consider the polynomials
\[
(g'_1, \ldots, g'_n, h'_1, \ldots, h'_{\rect{s(s-r)}}) = (z_1, \ldots, z_c, 1)
\left(
\begin{array}{cc}
D_xf' & D_yf' \\
w_1 \, \ldots \, w_n & 0
\end{array}
\right).
\]
Polynomials in $h'=(h'_1, \ldots, h'_{\rect{s(s-r)}})$ give the relations $z_i=0$,
for $i=\rect{(m-r)(s-r)+1}, \ldots, c$, and can be eliminated together with variables
$z_i, i=\rect{(m-r)(s-r)}+1, \ldots, c$. So the local equations of the Lagrange system
\eqref{eqlag} are:
\begin{equation}
\label{LocalLagSyst}
f' = 0, \qquad g' = 0.
\end{equation}
This is a square system consisting of $n+c$ polynomials in $n+c$ variables.

\subsection{Proof of Proposition \ref{prop:dimension}} \label{subsec:proof}

\begin{proof}[of Assertion 1 of Proposition \ref{prop:dimension}]
Let $\widetilde{\mathscr{W}}_{A,U} \subset \CC^n$ be the set defined by Lemma \ref{lemma:crit:points},
and $w \in \widetilde{\mathscr{W}}_{A,U}$. Then one has that all solutions $(x,y,z)$ to
\eqref{eqlag} (hence of the local version \eqref{LocalLagSyst}) satisfy $\rank \, A(x) = r$.
We deduce that there exists a $r \times r$ submatrix $N$ of $A(x)$ such that $\det \, N \neq 0$.
We prove below that there exists a non-empty Zariski open set $\mathscr{W}_{N,A,U} \subset \CC^n$ such
that, for $w \in \mathscr{W}_{N,A,U}$, the statement of Assertion 1 holds locally. Hence, to retrieve
the global property, it is sufficient to define $\mathscr{W}_{A,U}$ as the (finite) intersection of
sets $\widetilde{\mathscr{W}}_{A,U} \cap \mathscr{W}_{N,A,U}$, where $N$ varies in the collection of
$r \times r$ submatrices of $A$.

We suppose without loss of generality that $N$ is the upper-left $r \times r$ submatrix of $A$.
Let $(f',g')$ be the local Lagrange system defined in \eqref{LocalLagSyst}. Consider the polynomial
map
\[
  \begin{array}{lrcc}
  \varphi : &  \CC^{n+c} \times \CC^{n} & \longrightarrow & \CC^{n+c} \\
            &  (x,y,z,w) & \longmapsto & (f',g')
  \end{array}
\]
and, for a fixed $w \in \CC^n$, its section map
\[
  \begin{array}{lrcc}
  \varphi_{w} : &  \CC^{n+c} & \longrightarrow & \CC^{n+c} \\
            &  (x,y,z) & \longmapsto & \varphi(x,y,z,w)
  \end{array}.
\]
If $\varphi^{-1}(0) = \emptyset$, then for all $w \in \CC^n$, $\varphi_{w}^{-1}(0) = \emptyset$, and
the claim is proved by taking $\mathscr{W}_{N,A,U} = \widetilde{\mathscr{W}}_{A,U}$ (see Lemma
\ref{lemma:crit:points}).

Suppose now that $\varphi^{-1}(0) \neq \emptyset$ and let
$(x,y,z,w) \in \varphi^{-1}(0)$.  We claim that the Jacobian matrix of
$\varphi$ at $(x,y,z,w)$ has maximal rank. Hence, $0$ is a regular
value for $\varphi$ and by Thom's Weak Transversality Theorem
\cite[Proposition B.3]{SaSc13} there exist a non-empty Zariski open
set $\mathscr{W}_{N,A,U} \subset \CC^n$ such that for
$w \in \mathscr{W}_{N,A,U}$, $0$ is a regular value of
$\varphi_{w}$. This implies that, by the Jacobian criterion, the set
$\setZ(A,U,w) \cap \{(x,y,z) : \det N(x) \neq 0\}$ is empty or
zero-dimensional. We prove below this claim by exhibiting a
non-singular submatrix of $\jac \varphi$.

We remark that, since $f(A,U)$ satisfies $\sfG_2$, the Jacobian matrix $\jac f'$ has maximal rank at
$(x,y)$ and consider the submatrix of $\jac \varphi$ obtained by isolating:
\begin{itemize}
\item a non-singular maximal submatrix of $\jac f'$;
\item the derivatives of $g'_1, \ldots, g'_n$ with respect to $w_1, \ldots, w_n$,
  giving the identity block $\Id_n$.
\end{itemize}
The previous blocks define a submatrix of $\jac \varphi(x,y,z,w)$ of size $(n+c) \times (n+c)$
whose determinant does not vanish at $(x,y,z,w)$.
\end{proof}

\begin{proof}[of Assertion 2 of Proposition \ref{prop:dimension}]
  Let $w \in \CC^n$, and let $(x,y) \in \crit(\Pi_w,\setV_r)$. Since
  $A \in \zarA$, $f(A,U)$ satisfies $\sfG_2$ and $\setV_r(A,U)$ is
  smooth and equidimensional. Hence
  $(x,y) \in \reg(\setV_r(A,U))=\setV_r(A,U)$.  In particular
  $\jac f(x,y)$ has full rank. Since $(x,y) \in \crit(\Pi_w,\setV_r)$,
  the extended Jacobian matrix $\jac(f,\Pi_w)$ has a rank
  defect. Hence, there exists $z = (z_1,\ldots,z_{c+1})\neq 0$ with
  $z^T\jac(f,\Pi_w)(x,y) = 0$. If $z_{c+1}=0$, then
  $0 = z^T\jac(f,\Pi_w) = (z_1,\ldots,z_c)\jac f(x,y)$, which is a
  contradiction since $\jac f(x,y)$ has full rank. Then we can assume
  $z_{c+1}=1$, and hence that $(x,y,z) \in \setZ(A,U,w)$. We conclude
  that $\crit(\Pi_w,\setV_r)$ is contained in the projection of
  $\setZ(A,U,w)$ on $(x,y)$, as claimed.
\end{proof}

\section{Closure properties}\label{sec:closure}

The goal of this section is to prove Proposition \ref{prop:closure}.
We use notation of \cite[Section 5]{HNS2014}, which we recall
below.

{\it Notations.} For $M \in \GL_n(\CC)$, and $\setZ \subset \CC^n$ any set, we define
$$M^{-1}\setZ \coloneqq \{x \in \CC^n : Mx \in \setZ\} = \{M^{-1}x  : x \in \setZ\}.$$
  Remark that, if $w \neq 0$ and if $M \in \GL_n(\CC)$ with $w=M^{-1}_{(1)}$ (the first row of
  $M^{-1}$), then
  \begin{equation}
  \label{chvarw}
  \pi_1(M^{-1}\setZ) = M^{-1}_{(1)}\setZ = \pi_{M^{-1}_{(1)}}(\setZ) = \pi_w(\setZ).
  \end{equation}

  Let $\setZ \subset \CC^n$ be an algebraic variety of dimension $d$.
The $i-$equidimensional component of $\setZ$ is denoted by $\Omega_i(\setZ)$, $i=0, \ldots, d$.
We denote by $\scS(\mathcal Z)$ the union of the following sets:
\begin{itemize}
\item $\Omega_0(\mathcal Z) \cup \cdots \cup \Omega_{d-1}(\mathcal Z)$
\item the set $\sing(\Omega_d(\mathcal Z))$ of singular points of $\Omega_d(\mathcal Z)$
\end{itemize}
and by $\scC(\pi_i, \mathcal Z)$ the Zariski closure of the union of the
following sets:
\begin{itemize}
\item $\Omega_0(\mathcal Z) \cup \cdots \cup \Omega_{i-1}(\mathcal Z)$;
\item the union for $k \geq i$ of the sets $\crit(\pi_i, \reg(\Omega_k(\mathcal Z)))$ of
  critical points of the restriction of $\pi_i$ to the regular locus of $\Omega_k(\mathcal Z)$.
\end{itemize}
For $M \in \GL_n(\CC)$ we recursively define the collection
$\{\mathcal{O}_i(M^{-1}\setZ)\}_{0 \leq i \leq d}$ as follows:
\begin{itemize}
\item ${\mathcal O}_d(M^{-1}\setZ)=M^{-1}\setZ$;
\item ${\mathcal O}_i(M^{-1}\setZ)=\scS({\mathcal O}_{i+1}(M^{-1}\setZ))
\cup \scC(\pi_{i+1},  {\mathcal O}_{i+1}(M^{-1}\setZ)) \cup
\scC(\pi_{i+1},M^{-1}\setZ)$ for $i=0, \ldots, d-1$.
\end{itemize}

\new{We recall that an algebraic set $\mathcal{V}=\zeroset{I} \subset \CC^n$ is in Noether
position with respect to variables $x_1, \ldots, x_i$, if and only if the morphism
$\varphi \colon \CC[x_1, \ldots, x_i] \hookrightarrow \bigslant{\CC[x_1, \ldots, x_n]}{I}$
is injective and integral. If this is the case, the induced morphism $\varphi^* \colon \mathcal{V} 
\to \CC^i$ is the projection $\varphi^*(x) = (x_1, \ldots, x_i)$ and $\varphi^*$ is one-to-one.
}

The following two properties have been defined in \cite[Section 5]{HNS2014}.  

{\it Property $\sfP(\setZ)$.} Let ${\setZ} \subset \CC^n$
be an algebraic set of dimension $d$. We say that $M \in \GL_n(\CC)$
satisfies $\sfP(\setZ)$ when for all $i = 0, 1, \ldots, d$
\begin{enumerate}
\item ${\mathcal O}_i(M^{-1}{\setZ})$ has dimension $\leq i$; 
\item ${\mathcal O}_i(M^{-1}{\setZ})$ is in Noether position with respect to $x_1, \ldots, x_i$. 
\end{enumerate}


{\it Property ${\sf Q}(\setZ)$.} Let $\setZ$ be an algebraic set of
dimension $d$ and $1 \leq i \leq d$. We say that ${\sf Q}_i(\setZ)$ holds if
for any connected component $\cc$ of ${\setZ}\cap \RR^n$ the boundary of
$\pi_i(\cc)$ is contained in $\pi_i({\mathcal O}_{i-1}(\setZ) \cap {\cc})$.
{We say that ${\sf Q}(\setZ)$ holds if ${\sf Q}_i(\setZ)$ holds for
all $1 \leq i \leq d$.}

In \cite{HNS2014} the authors proved that given any algebraic variety $\setZ$ of dimension $d$, Property
${\sf P}(\setZ)$ holds generically in $\GL_n( \CC)$ (Proposition 17) and that if $M \in \GL_n(\CC)$
satisfies ${\sf P}(\setZ)$, then ${\sf Q}(M^{-1}\setZ)$ holds (Proposition 18).
We use these results in the following proof of Proposition \ref{prop:closure}.

\begin{proof}[of Assertion 1 of Proposition \ref{prop:closure}]
  \mynew{Let $A$ and $U$ be as in the hypothesis.}  Let
  $\zarM\mynew{_{A,U}} \subset \GL_n(\CC)$ be the non-empty Zariski
  open set defined in \cite[Proposition 17]{HNS2014} for
  $\setZ = \setD_r$. \mynew{This set might depend on the choice of $A$
    and $U$, as explicited by the indices.}  One obtains that every
  $M \in \zarM\mynew{_{A,U}}\cap \GL_n(\QQ)$ satisfies
  ${\sfP}(\setD_r)$. Remark that since $M \in \GL_n(\QQ)$, there is a
  natural bijective correspondence between the set of connected
  components of $\setD_r \cap \RR^n$ and the ones of
  $M^{-1}\setD_r \cap \RR^n$ given by $\cc \leftrightarrow
  M^{-1}\cc$.
  Fix a connected component
  $M^{-1}\cc \subset M^{-1}\setD_r \cap \RR^n$ and consider the
  projection $\pi_i$ restricted to $M^{-1}\setD_r \cap \RR^n$. Since
  $M \in \zarM\mynew{_{A,U}}$, by \cite[Proposition 18]{HNS2014} the
  boundary of $\pi_i(M^{-1}\cc)$ is contained in
  $\pi_i({\mathcal O}_{i-1}(M^{-1}\setD_r) \cap M^{-1}\cc)$ and in
  particular in $\pi_i(M^{-1}\cc)$. This implies that
  $\pi_i(M^{-1}\cc)$ is closed \mynew{for all $i$. Let
    $\varphi \colon \CC^{n \times n} \to \CC^n$ be the map sending
    $M \in \CC^{n \times n}$ to its first row $M_{(1)}$.  The map
    $\varphi$ is a projection, hence a morphism. Applying \cite[9.1,
    Ex.III]{hartshorne}, one gets that
    $\varphi(\zarM_{A,U})\coloneqq \zarW^{(1)}_{A,U}$ is Zariski-open
    in $\CC^n$.  In particular, for $i=1$ and applying \eqref{chvarw},
    one gets that $\pi_w(\cc)$ is closed, for
    $w \in \zarW^{(1)}_{A,U}$}.
\end{proof}


\begin{proof}[of Assertion 2 of Proposition \ref{prop:closure}] 

  \mynew{Let $A$ and $U$ be as in the hypothesis. Fix $0 \leq p \leq r$.
    Since $A \in \zarA$, by Proposition \ref{prop:regularity} the set $\setD_p$
    has the expected co-dimension, hence $c = (m-p)(s-p)$. Let $\eta_1, \ldots,
    \eta_c$ be local equations for $\setD_p$.
    
    We consider the following constructible set:
    $$
    J = \left\{(x,z,w) : x \in \setD_p \cap {\mathcal{N}_r(A,U)},
    w_i = \sum_j z_j \frac{\partial \eta_j}{\partial x_i}, z \neq 0\right\}
    $$
    and the projection $\pi : \CC^{N} \to \CC^n$, $\pi(x,y,z,w) = w$, with
    $N=2n+c$, where $c$ is the co-dimension of $\setD_p$.
    We claim that $\pi(J)$ is a constructible subset of $\CC^n$ of positive
    co-dimension. Hence there exists a non-empty Zariski open set
    $\zarW^{(2)}_{A,U} \subset (\CC^n \setminus \pi(J))$ such that Assertion 2
    holds for $w \in \zarW^{(2)}_{A,U}$.
    We conclude the proof by defining $\zarW'_{A,U} = \zarW^{(1)}_{A,U} \cap
    \zarW^{(2)}_{A,U}$ (where $\zarW^{(1)}_{A,U}$ has been defined in the proof
    of Assertion 1). 

    We prove now our claim.
    By Assertion (2b) of Proposition \ref{prop:regularity}, the set
    $\setD_p \cap \mathcal{N}_r(A,U)$ has positive codimension in
    $\setD_p$, hence it has codimension at least $c+1$. Moreover, the
    equations $w_i = \sum_j z_j \frac{\partial \eta_j}{\partial x_i}$
    define an algebraic set of co-dimension $n$ (indeed, their
    Jacobian matrix has full rank, having a $n \times n$ identity
    matrix corresponding to the derivatives w.r.t. $w_i$).  We deduce
    that $J$ is a constructible set of dimension at most
    $N-(c+1+n) = n-1$ in $\CC^N$.  Hence its projection $\pi(J)$ has
    dimension at most $n-1$ in $\CC^n$, as claimed.  }
%
\end{proof}

\section{Experiments} \label{sec:experiments}

This section reports on experiments made with a first implementation of our algorithm. Note that for computing rational parametrizations, we use Gr\"obner bases and change of ordering algorithms \cite{fglm}. Our experiments are done using the C library \textsc{FGb}, developed by J.-C. Faug\`ere \cite{faugere2010fgb} and interfaced with {\sc Maple}.

We start by comparing our implementation with implementations of
general algorithms based on the critical point method in {\sc RAGlib}
\cite{raglib}. Next, we comment the behaviour of our algorithm on
special examples that are well-known by the research community working on
linear matrices.

\subsection{Comparison with {\sc RAGlib}}


We have generated randomly linear matrices for various values of $m=s$
(\rect{for simplicity we perform computations on square matrices})
and $n$ and run our implementation for different values of $r$.
By randomness of rational numbers we mean that we generate
couples of integers chosen with uniform distribution in a fixed interval
(here $[-100,100]$).

Our implementation is in \textsc{Maple} and computations have been performed
on an Intel(R) Xeon(R) CPU
$E7540@2.00{\rm GHz}$ 256 Gb of RAM.  We report in Table
\ref{tab:dense} numerical data of our tests. The reader can find an implementation
of our algorithm and the details of these computations at the following permanent link
\begin{center}
\url{https://www.unilim.fr/pages_perso/simone.naldi/lowrank.html}
\end{center}
For any choice of $m=s$,
$2 \leq r \leq m-1$ and $n$, we generate a random dense linear
matrix $A$ and we let {\sf LowRank} run with input $(A,r)$.
\begin{table}[!ht]
\centering
{\scriptsize
\begin{tabular}{|l|rrrr|l|rrrr|}
\hline
$(m,r,n)$ & {\sf PPC} & {\sf LowRank} & {\sf deg} & {\sf maxdeg} & $(m,r,n)$ & {\sf PPC} & {\sf LowRank} & {\sf deg} & {\sf maxdeg} \\
\hline
\hline
$(3,2,2)$ & 0.2      & 6     & 9    & 6  &   $(5,2,3)$  & 0.9 & 0.5  & - & -\\
$(3,2,3)$ & 0.3      & 7.5    & 21   & 12   & $(5,2,4)$  & 1 & 0.5  & - & - \\
$(3,2,4)$ & 0.9      & 9.5    & 33   & 12   & $(5,2,5)$ & 1.6   & 0.5  & -  & - \\
$(3,2,5)$ & 5.1      & 13.5    & 39   & 12   & $(5,2,6)$ & 3   & 0.6  & -  & -  \\
$(3,2,6)$ & 15.5     & 15    & 39   & 12   &  $(5,2,7)$ & 4.2   & 0.7  &  - & -  \\
$(3,2,7)$ & 31       & 16.5    & 39   & 12   & $(5,2,8)$ &  8  & 0.7  & -  & -  \\
$(3,2,8)$ & 109      & 18    & 39   & 12   &   $(5,2,9)$ & $\infty$   & 903   & 175  & 175  \\
$(3,2,9)$ & 230      & 20    & 39   & 12   & $(5,3,2)$  & 0.4 & 0.5  & - & - \\
$(4,2,2)$ & 0.2      & 0.5     & -    & -    & $(5,3,3)$  & 0.5 & 0.5  & - & - \\
$(4,2,3)$ & 0.3      & 0.5    & -    & -    &  $(5,3,4)$  & 43 & 22  & 50 & 50   \\ 
$(4,2,4)$ & 2.2      & 2.5     & 20   & 20   & $(5,3,5)$ & $\infty$ &  5963  & 350  & 300  \\
$(4,2,5)$ & 12.2     & 26    & 100  & 80   &  $(5,4,2)$  & 0.5 & 125  & 25 & 20 \\
$(4,2,6)$ & $\infty$ & 593    & 276  & 176  &  $(5,4,3)$  & 10 & 167  & 105 & 80  \\
$(4,2,7)$ & $\infty$ & 6684   & 532  & 256  &  $(5,4,4)$  & $\infty$ & 561  & 325 & 220 \\
$(4,2,8)$ & $\infty$ & 42868  & 818  & 286  & $(5,4,5)$  & $\infty$ & 5574  & 755 & 430    \\
$(4,2,9)$ & $\infty$ & 120801 & 1074 & 286  & $(6,3,3)$  & 4 & 1  & - & -  \\
$(4,3,3)$ & 1        & 8    & 52   & 36   &  $(6,3,4)$  & 140 & 1  & - & -   \\
$(4,3,4)$ & 590      & 18    & 120  & 68   &  $(6,3,5)$  & $\infty$ & 1  & - & -  \\ 
$(4,3,5)$ & $\infty$ & 56    & 204  & 84   &  $(6,3,6)$  & $\infty$ & 2  & - & -   \\ 
$(4,3,6)$ & $\infty$ & 114    & 264  & 84   & $(6,3,7)$  & $\infty$ & 2  & - & -  \\ 
$(4,3,7)$ & $\infty$ & 124    & 284  & 84   & $(6,3,8)$  & $\infty$ & 2  & - & -  \\  
$(4,3,8)$  & $\infty$ & 124  & 284 & 84   &  $(6,4,2)$  & 0.6 & 40 & - & -     \\
$(4,3,9)$  & $\infty$  & 295  & 284 & 84  &  $(6,4,3)$  & 1 & 64 & - & -  \\
$(4,3,10)$ & $\infty$ & 303  & 284 & 84 &    $(6,4,4)$  & 341  & 300  & 105 & 105   \\ 
$(4,3,11)$ & $\infty$ & 377  & 284 & 84 &  $(6,5,3)$  & 95  & 276  & 186 & 150\\
$(5,2,2)$  & 0.6  & 0.5  & - & - & $(6,5,4)$ & $\infty$ & 8643  & 726 & 540 \\
\hline
\end{tabular}
}
\caption{Timings and degrees for dense linear matrices}
\label{tab:dense}
\end{table}
We compare our algorithm ({timings in seconds} reported in column ``{\sf LowRank}'') with the function {\sf PointsPerComponents} (column ``{\sf PPC}'') of {\sf RAGlib} \cite{raglib}.
The input of {\sf PointsPerComponents} is the list of all $(r+1) \times (r+1)$ minors of the matrix $A(x)$ (generating the ideal of the set of matrices of rank $\leq r$ in the pencil $A(x)$).

The symbol $\infty$ means that no result has been returned after $4$ days of computation. \new{In column {\sf deg} we report the degree of $q_{n+1}$ in the output parametrization. Finally, in column {\sf maxdeg} we report the maximum of the degrees of the partial rational parametrizations (corresponding to the different recursive steps of {\sf LowRank}). When {\sf deg} is ``-'', we mean that the empty list is returned}.

We make the following remarks about Table \ref{tab:dense}.
\begin{enumerate}
\item We first observe that our algorithm is most of the time faster
  than {\sf RAGlib} and it allows to tackle examples that are out of
  reach of {\sf RAGlib}. 
\item The growth in terms of timings with respect to $n$ seems to
  respect the corresponding growth in terms of degrees of output
  parametrizations; in particular note that we have established that
  for $r$ and $m$ fixed, the sum of the degrees of parametrizations we
  need to compute stabilizes when $n$ grows. This is observed in
  practice of course and is reflected in our timings compared to those
  of {\sf RAGlib}.
\item Accordingly to the related Multilinear B\'ezout Bounds computed in
  section \ref{ssec:compl:mbb}, the degrees of rational parametrizations
  stabilize when $n$ grows, since when $n>\rect{(m+r)(s-r)}$ and the input is generic,
  {\sf LowRank} does not compute critical points at first calls. This fact
  is remarkable, since:
  \begin{itemize}
  \item it is known (see \cite[Ch.\,II]{arbarello2011geometry}) that
    a natural geometric invariant associated to $\setD_r$, its degree
    as complex algebraic set, does not depend on the dimension $n$ of the
    affine section (one can prove easily that generically this degree is
    given by the Thom-Porteous-Giambelli formula,
    {\it cf.} \cite[Ch.\,II,\S\,4]{arbarello2011geometry});
  \item an algebraic invariant naturally associated to the output-size
    (the degree of $q_{n+1}$) is constant in $n$, coherently with the aforementioned
    geometric invariant.
  \end{itemize}
\end{enumerate}

Finally, we give a final remark on potential {\it a posteriori} verification
of the correctness of the output of {\sf LowRank}. Deciding whether a finite
set, encoded by a rational parametrization, meets every connected component
of a given real algebraic set, is a hard problem, far from being solved, 
both from a theoretical and computational viewpoint. As far as the authors
know, there are no symbolic or numerical algorithms able to perform this task.
Also, producing such a certificate seems to be hard to imagine, but this
was not among the goals of this paper.
\mynew{In the recent paper \cite{SaSc13}, an algorithm to address connectivity
  queries (for instance, deciding whether $2$ points of a smooth compact real
  algebraic set lie on the same connected component) has been developed.}


\subsection{Examples}

In this last section, we consider some examples of linear matrices
coming from the literature, and we test the behavior of {\sf LowRank}.
We consider examples of symmetric linear matrices since, as observed
in Section \ref{motiv}, the main motivation for solving the
real root finding problem is to obtain dedicated algorithms for
spectrahedra and semidefinite programming. All these tests can be
replicated via the Maple library {\sc Spectra} via the function $SolveLMI$
(cf. the documentation of Spectra and \cite{spectra} for details).

\begin{example}[The Cayley cubic] \label{ex:cayley}
We consider the $3 \times 3$ linear matrix
\[
A(x) =
\left(
\begin{array}{ccc}
1 & x_1 & x_2 \\
x_1 & 1 & x_3 \\
x_2 & x_3 & 1
\end{array}
\right).
\]


The convex region $\{x \in \RR^3 \ \big| \ A(x) \succeq
0\}$ is the Cayley spectrahedron. We run
our algorithm
with input $(A,r)$ with $r=2$ and $r=1$ (the case $r=0$
is trivial since $A(x)$ is always non-zero and hence
$\setD_0$ is empty). In both cases, the algorithm first
\mynew{chooses a random matrix $U$, then} verifies that
the genericity assumptions are satisfied.

Let us first analyze the case $r=2$. Our algorithm
runs $3$ recursive steps.
{Its output is a rational parametrization of degree
  $14$ with $12$ real solutions and $2$ complex solutions. We give
  below details of each recursive call of {\sf LowRankRec}. At the
  first, at step \ref{new:rec:5}, a rational parametrization of
  degree $5$ is computed, with the following $5$ real solutions:}
\[
\left\{\left(\begin{array}{c}1\\1\\1\end{array}\right), \:
\left(\begin{array}{c}1\\-1\\-1\end{array}\right), \:
\left(\begin{array}{c}-1\\1\\-1\end{array}\right), \:
\left(\begin{array}{c}-1\\-1\\1\end{array}\right), \:
\left(\begin{array}{c} 18.285118452 \\ 164.322822823 \\ 4.552268485 \end{array}\right)
\right\}.
\]
The coordinates of the fifth point are approximated to
$9$ certified digits and such approximation can be computed
by isolating the coordinates in intervals of rational numbers as:
\[
\begin{array}{l}
x_1 \in [\frac{21081306277346124211}{1152921504606846976}, \frac{21081306277346754459}{1152921504606846976}]
\approx  18.285118452 \\[.3em]
x_2 \in [\frac{5920353629066611305}{36028797018963968}, \frac{23681414516266799197}{144115188075855872}]
\approx  164.322822823 \\[.3em]
x_3 \in [\frac{10496816461511385723}{2305843009213693952}, \frac{2624204115377866059}{576460752303423488}]
\approx 4.552268485.
\end{array}
\]

Remark that it also computes the $4$ singular points of $\setD_2$, where
the rank of $A$ is $1$. At the second (resp. third) recursive call, it
computes a rational parametrization of degree $6$ (resp. of degree $3$)
with $4$ (resp. $3$) real solutions. 

In the case $r=1$, step \ref{new:rec:4} of {\sf LowRankRec} returns a
rational parametrization of degree $4$ which encodes the $4$ singular
points of $\setD_2 \cap \RR^3$, that is $\setD_1 \cap \RR^3$. {At the
second and third recursions, {\sf LowRankRec} returns empty lists.}


{We finally remark that the above results are typical, in the sense
  that the $4$ singular points contained in $\setD_1 \cap \RR^3$ are
  always computed at the first recursion step, both in case $r=2$ and
  $r=1$.  Conversely, the coordinates of the other real solutions
  depend on the choice of random parameters (while their number is
  constant).  Moreover, all computations end after a few seconds ($< 5$
  s).}

\end{example}


\begin{example} \label{ex:berk}
Let
\[
A(x) =
\left(
\begin{array}{cccc}
a_{1} & x_1 & x_2 & x_3 \\
x_1 & a_{2} & x_3 & x_4 \\
x_2 & x_3 & a_{3} & x_5 \\
x_3 & x_4 & x_5 & a_{4}
\end{array}
\right),
\]
where $x=(x_1,x_2,x_3,x_4,x_5)$ are variables and $(a_{1},a_{2},a_{3},a_{4})$
are parameters. This example shows that the genericity assumptions of our
algorithm are satisfied up to small perturbation of the input data.

We let $(a_{1},a_{2},a_{3},a_{4})$ vary randomly in $\QQ^4$. For all random
instances, we observe that the inputs $(A,3),(A,2)$ and $(A,1)$ verify the genericity
assumptions, and that the degrees of the rational parametrizations
returned at each recursion step are constant, while the number of
real solutions changes with parameters. We summarize our results in Table
\ref{tab:example:berkeley}.

\begin{table}[!ht]
  \centering
  \begin{tabular}{|c||c|c|c|}
    \hline
     & $r=3$ & $r=2$ & $r=1$  \\
    \hline
    \hline
    partial degrees   & [12 24 24 12 4]   & [12 20 8 0 0]  & [0 0 0 0 0]  \\
    \hline
    total degree   & 76   & 40  & 0  \\
    \hline
    time (s)  & 768  & 21.5  & 4.6  \\
    \hline
  \end{tabular}
  \vspace{0.4cm}
  \caption{Degrees and timings for Example \ref{ex:berk} with generic parameters}
  \label{tab:example:berkeley}
\end{table}


\end{example}

\begin{example}[The pillow] \label{ex:pillow}
Let
\[
A(x) =
\left(
\begin{array}{cccc}
1 & x_1 & 0 & x_1 \\
x_1 & 1 & x_2 & 0 \\
0 & x_2 & 1 & x_3 \\
x_1 & 0 & x_3 & 1
\end{array}
\right).
\]
The spectrahedron $S = \{x \in \RR^3 : A(x) \succeq 0\}$ is known
as {\it the pillow}.
%

The Zariski closure of
its boundary is the real trace of the complex hypersurface defined by the vanishing of
\[
\det A(x) = 1-x_3^2-x_2^2-2x_1^2+x_1^2x_3^2-2x_1^2x_2x_3+x_1^2x_2^2.
\]

We tested our algorithm with input $(A,2)$. We obtain
that at the first recursion, at step \ref{new:rec:4}
a rational parametrization $q=(q_0,q_1,q_2,q_3,q_4)$
of degree $4$ (with only real roots)
is computed. By isolating the $4$ real roots of $q_4$
as in Example \ref{ex:cayley}, one gets the following rational
approximation of one of the singular points:
\[
\begin{array}{l}
x_1\in [-\frac{6521908912666475339}{9223372036854775808}, -\frac{13043817825332644843}{18446744073709551616}] \approx -{\sqrt{2}}/{2} \\ [.3em]
x_2\in [\frac{26087635650665343561}{36893488147419103232}, \frac{6521908912666428733}{9223372036854775808}] \approx {\sqrt{2}}/{2} \\[.3em]
x_3\in [-\frac{6521908912666412349}{9223372036854775808}, -\frac{13043817825332731855}{18446744073709551616}] \approx -{\sqrt{2}}/{2}.
\end{array}
\]



\end{example}


\begin{acknowledgements}
  The second author is supported by the FMJH Program PGMO and would like to thank LAAS-CNRS and the Mathematics Department of Technische Universit\"at Dortmund where he was working when a significant part of this work was designed. The third author is supported by Institut Universitaire de France. We thank the reviewers for their help in improving the first version of this paper.
\end{acknowledgements}

\def\cfac#1{\ifmmode\setbox7\hbox{$\accent"5E#1$}\else
  \setbox7\hbox{\accent"5E#1}\penalty 10000\relax\fi\raise 1\ht7
  \hbox{\lower1.15ex\hbox to 1\wd7{\hss\accent"13\hss}}\penalty 10000
  \hskip-1\wd7\penalty 10000\box7}

\end{document}